\documentclass[11pt]{article}

\usepackage{amsmath}
\usepackage{amsfonts}
\usepackage{amssymb}
\usepackage{amsthm}
\usepackage[usenames]{color}
\usepackage{fullpage}
\usepackage{xspace}
\usepackage{url,ifthen}
\usepackage[margin=1in]{geometry}
\usepackage{xspace}
\usepackage{graphicx}
\usepackage[usenames]{color}
\definecolor{DarkGreen}{rgb}{0.1,0.5,0.1}
\definecolor{DarkRed}{rgb}{0.5,0.1,0.1}
\definecolor{DarkBlue}{rgb}{0.1,0.1,0.5}
\usepackage[small]{caption}
\usepackage{float}
\usepackage{balance}

\newtheorem{theorem}{Theorem}[section]
\newtheorem*{namedtheorem}{\theoremname}
\newcommand{\theoremname}{testing}

\newtheorem{lemma}[theorem]{Lemma}

\newtheorem{claim}[theorem]{Claim}

\newtheorem{prop}[theorem]{Proposition}
\newtheorem{fact}[theorem]{Fact}

\newtheorem{corollary}[theorem]{Corollary}
\newtheorem{condition}[theorem]{Condition}

\newtheorem*{question*}{Question}

\theoremstyle{plain}

\theoremstyle{definition}
\newtheorem{definition}[theorem]{Definition}

\newtheorem{remark}[theorem]{Remark}

\usepackage[backref,colorlinks,bookmarks=true]{hyperref}

\newcommand{\ignore}[1]{}
\newcommand{\dimt}[2]{\dim_{#1}^{\tau}\left( #2 \right)}

\renewcommand{\Pr}{\mathop{\bf Pr\/}}                    % should we change these to \mathbb for consistency of single-letter functionals
\newcommand{\E}{\mathop{\bf E\/}}

\newcommand{\poly}{\mathrm{poly}}

\newcommand{\xtil}{\widetilde{x}}
\newcommand{\ytil}{\widetilde{y}}
\newcommand{\util}{\widetilde{x}}

\newcommand{\iprod}[1]{\langle #1 \rangle}
\newcommand{\norm}[1]{\lVert #1 \rVert}
\newcommand{\ha}[1]{\widetilde{#1}}

\DeclareMathOperator{\spn}{span}

\newcommand{\R}{\mathbb R}

\newcommand{\N}{\mathbb N}

\newcommand{\eps}{\varepsilon}

\newcommand{\ot}{\otimes}

\newcommand{\calM}{\mathcal{M}}
\newcommand{\calN}{\mathcal{N}}

\newcommand{\calS}{\mathcal{S}}

\newcommand{\calU}{\mathcal{U}}
\newcommand{\calV}{\mathcal{V}}
\newcommand{\calW}{\mathcal{W}}

\newcommand{\kr}[1]{\text{Krank}_{#1}}
\newcommand{\vect}{\text{vec}}

\newcommand{\abs}[1]{\left\lvert #1 \right\rvert}

\newcommand{\set}[1]{\left\{ #1 \right\}}

\newcommand{\Rstz}[2]{#1^{*#2}}

\newcommand{\extz}[2]{\text{Ext}_{*#1}\left( #2 \right)}
\newcommand{\transpose}[1]{{#1}^{\sc T}}

\newcommand{\Alpha}{\mathbf{\alpha}}
\newcommand{\Beta}{\mathbf{\beta}}
\newcommand{\Util}{\widetilde{U}}
\newcommand{\Vtil}{\widetilde{V}}

\newcommand{\krp}{\odot}
\newcommand{\spc}[2]{{#1}^{(#2)}}

\makeatletter
\floatstyle{ruled}
\newfloat{fragment}{H}{lop}
\floatname{fragment}{Algorithm}
\renewcommand{\floatc@ruled}[2]{\vspace{2pt}{\@fs@cfont \#1.\:} \#2 \par
 \vspace{1pt}}
\makeatother

\renewenvironment{proof}[1][\proofname]{\medskip\noindent{\bfseries #1: }}{$\blacksquare$\vskip \belowdisplayskip}

\title{Smoothed Analysis of Tensor Decompositions}

\author{Aditya Bhaskara\thanks{Google Research NYC. Email: \textsf{bhaskara@cs.princeton.edu}. Work done while the author was at EPFL, Switzerland.} \and Moses Charikar\thanks{Princeton University. Email: \textsf{moses@cs.princeton.edu}. Supported by NSF awards CCF 0832797, AF 1218687 and CCF 1302518} \and Ankur Moitra\thanks{Massachusetts Institute of Technology, Department of Mathematics and CSAIL. Email: \textsf{moitra@mit.edu}. Part of this work was done while the author was a postdoc at the Institute for Advanced Study and was supported in part by NSF grant No.DMS-0835373 and by an NSF Computing and Innovation Fellowship.} \and Aravindan Vijayaraghavan\thanks{Carnegie Mellon University. Email: \textsf{aravindv@cs.cmu.edu}. Supported by the Simons Postdoctoral Fellowship.}}
\begin{document}
\date{}
\maketitle

\setcounter{page}{0} \thispagestyle{empty}

\begin{abstract}

Low rank decomposition of tensors is a powerful tool for learning generative models. The uniqueness of decomposition gives tensors a significant advantage over matrices. However, tensors pose significant algorithmic challenges and tensors analogs of much of the matrix algebra toolkit are unlikely to exist because of hardness results. Efficient decomposition in the overcomplete case (where rank exceeds dimension) is particularly challenging. We introduce a smoothed analysis model for studying these questions and develop an efficient algorithm for tensor decomposition in the highly overcomplete case (rank polynomial in the dimension). In this setting, we show that our algorithm is robust to inverse polynomial error -- a crucial property for applications in learning since we are only allowed a polynomial number of samples. While algorithms are known for exact tensor decomposition in some overcomplete settings, our main contribution is in analyzing their stability in the framework of smoothed analysis. 

Our main technical contribution is to show that tensor products of perturbed vectors are linearly independent in a robust sense (i.e. the associated matrix has singular values that are at least an inverse polynomial). This key result paves the way for applying tensor methods to learning problems in the smoothed setting. In particular, we use it to obtain results for learning multi-view models and mixtures of axis-aligned Gaussians where there are many more ``components'' than dimensions. The assumption here is that the model is not adversarially chosen, formalized by a perturbation of model parameters. We believe this an appealing way to analyze realistic instances of learning problems, since this framework allows us to overcome many of the usual limitations of using tensor methods. 

\end{abstract}

\newpage

\section{Introduction}

\subsection{Background}

Tensor decompositions play a central role in modern statistics (see e.g. \cite{M}). To illustrate their usefulness, suppose we are given a matrix $M = \sum_{i=1}^R a_i \otimes b_i$ When can we {\em uniquely} recover the factors $\{a_i\}_i$ and $\{b_i\}_i$ of this decomposition given access to $M$? In fact, this decomposition is almost never unique (unless we require that the factors $\{a_i\}_i$ and $\{b_i\}_i$ are orthonormal, or that $M$ has rank one). But given a tensor $T = \sum_{i=1}^R a_i \otimes b_i \otimes c_i$ there are general conditions under which $\{a_i\}_i$, $\{b_i\}_i$ and $\{c_i\}_i$ are uniquely determined (up to scaling) given $T$; perhaps the most famous such condition is due to Kruskal \cite{Kru}, which we review in the next section. 

Tensor methods are commonly used to establish that the parameters of a generative model can be identified given third (or higher) order moments. In contrast, given just second-order moments (e.g. $M$) we can only hope to recover the factors up to a rotation. This is called the {\em rotation problem} and has been an important issue in statistics since the pioneering work of psychologist Charles Spearman (1904) \cite{S}. Tensors offer a path around this obstacle precisely because their decompositions are often unique, and consequently have found applications in phylogenetic reconstruction \cite{C}, \cite{MR}, hidden markov models \cite{MR}, mixture models \cite{HK}, topic modeling \cite{AFHKL}, community detection \cite{AGHK}, etc. 

However most tensor problems are hard: computing the rank \cite{H}, the best rank one approximation \cite{HL} and the spectral norm \cite{HL} are all $NP$-hard. Also many of the familiar properties of matrices do not generalize to tensors. For example, subtracting the best rank one approximation to a tensor can actually increase its rank \cite{SC} and there are rank three tensors that can be approximated arbitrarily well by a sequence of rank two tensors. One of the few algorithmic results for tensors is an algorithm for computing tensor decompositions in a restricted case. Let $A, B$ and $C$ be matrices whose columns are $\{a_i\}_i$, $\{b_i\}_i$ and $\{c_i\}_i$ respectively. 

\begin{theorem} \cite{LRA}, \cite{C}
If $\mbox{rank}(A) = \mbox{rank}(B) = R$ and no pair of columns in $C$ are multiples of each other, then there is a polynomial time algorithm to compute the minimum rank tensor decomposition of $T$. Moreover the rank one terms in this decomposition are unique (among all decompositions with the same rank). 
\end{theorem}

\noindent If $T$ is an $n \times n \times n$ tensor, then $R$ can be at most $n$ in order for the conditions of the theorem to be met. This basic algorithm has been used to design efficient algorithms for phylogenetic reconstruction \cite{C}, \cite{MR}, topic modeling \cite{AFHKL}, community detection \cite{AGHK} and learning hidden markov models and mixtures of spherical Gaussians \cite{HK}. However algorithms that make use of tensor decompositions have traditionally been limited to the full-rank case, and our goal is to develop {\em stable} algorithms that work for $R = \mbox{poly}(n)$. Recently Goyal et al \cite{GVX} gave a robustness analysis for this decomposition, and we give an alternative proof in Appendix~\ref{asec:stable}.  

In fact, this basic tensor decomposition can be bootstrapped to work even when $R$ is larger than $n$ (if we also increase the order of the tensor). The key parameter that dictates when one can efficiently find a tensor decomposition (or more generally, when it is unique) is the {\em Kruskal rank}:

\begin{definition}
The Kruskal rank (or $\kr{}$) of a matrix $A$ is the largest $k$ for which every set of $k$ columns are linearly independent. Also the $\tau$-{\em robust k-rank} is denoted by $\kr{\tau}(A)$, and is the largest $k$ for which every $n \times k$ sub-matrix $A_{|S}$ of $A$ has $\sigma_k (A_{|S}) \ge 1/\tau$.
\end{definition}

\noindent How can we push the above theorem beyond $R = n$? We can instead work with an order $\ell$ tensor. To be concrete set $\ell = 5$ and suppose $T$ is an $n \times n \times ... \times n$ tensor. We can ``flatten" $T$ to get an order three tensor $$ T = \sum_{i = 1}^R \underbrace{A^{(1)}_i \otimes A^{(2)}_i}_{\mbox{factor}} \otimes \underbrace{A^{(3)}_i \otimes A^{(4)}_i}_{\mbox{factor}} \otimes \underbrace{A^{(5)}_i}_{\mbox{factor}}$$ Hence we get an order three tensor $\widehat{T}$ of size $n^2 \times n^2 \times n$. Alternatively we can define this ``flattening" using the following operation:

\begin{definition}
The Khatri-Rao product of $U$ and $V$ which are size $m \times r$ and $n \times r$ respectively is an $mn \times r$ matrix $U \odot V$ whose $i^{th}$ column is $u_i \otimes v_i$. 
\end{definition}

Our new order three tensor $\widehat{T}$ can be written as: $$\widehat{T} = \sum_{i = 1}^R \left (A^{(1)} \odot A^{(2)} \right )_i \otimes \left (A^{(3)} \odot A^{(4)} \right )_i \otimes A^{(5)}$$ The factors are the columns of $A^{(1)} \odot A^{(2)}$, the columns of $A^{(3)} \odot A^{(4)}$ and the columns of $A^{(5)}$. The crucial point is that the Kruskal rank of the columns of $A^{(1)} \odot A^{(2)}$ is in fact at least the sum of the Kruskal rank of the columns of $A^{(1)}$ and $A^{(2)}$ (and similarly for $A^{(3)} \odot A^{(4)}$) \cite{AMR}, \cite{BCV}, but this is tight in the worst-case. Consequently this ``flattening" operation allows us use the above algorithm unto $R = 2n$; since the rank ($R$) is larger than the largest dimension ($n$), this is called the {\em overcomplete} case. 

Our main technical result is that in a natural {\em smoothed analysis} model, the Kruskal rank {\em robustly multiplies} and this allows us to give algorithms for computing a tensor decomposition even in the highly overcomplete case, for any $R = \mbox{poly}(n)$ (provided that the order of the tensor is large \-- but still a constant). Moreover our algorithms have immediate applications in learning mixtures of Gaussians and multiview mixture models. 

\subsection{Our Results}

We introduce the following framework for studying tensor decomposition problems:

\begin{itemize}

\item An adversary chooses a tensor $T = \sum_{i=1}^R A^{(1)}_i \otimes A^{(2)}_i \otimes ... \otimes A^{(\ell)}_i$.

\item Each vector $a^j_i$ is $\rho$-perturbed to yield $\tilde{a}^j_i$.\footnote{An (independent) random gaussian with zero mean and variance $\rho^2/n$ in each coordinate is added to $a^j_i$ to obtain $\tilde{a}^j_i$. We note that we make the Gaussian assumption for convenience, 
%but our analysis uses only basic anti-concentration properties of the perturbations.
but our analysis seems to apply to more general perturbations.
}

\item We are given $\widetilde{T} = \sum_{i=1}^R \widetilde{A}^{(1)}_i \otimes \widetilde{A}^{(2)}_i \otimes ... \otimes \widetilde{A}^{(\ell)}_i$ (possibly with noise.)

\end{itemize}

\noindent Our goal is to recover the factors $\{\widetilde{A}^{(1)}_i\}_i$, $\{\widetilde{A}^{(2)}_i\}_i$, ..., $\{\widetilde{A}^{(\ell)}_i\}_i$ (up to rescaling). This model is directly inspired by {\em smoothed analysis} which was introduced by Spielman and Teng \cite{ST1}, \cite{ST2} as a framework in which to understand why certain algorithms perform well on realistic inputs. 

In applications in learning, tensors are used to encode low-order moments of the distribution. In particular, each factor in the decomposition represents a ``component". The intuition is that if these ``components" are not chosen in a worst-case configuration, then we can obtain vastly improved learning algorithms in various settings. For example, as a direct consequence of our main result, we will give new algorithms for learning mixtures of spherical Gaussians again in the framework of smoothed analysis (without any additional separation conditions). There are no known polynomial time algorithms to learn such mixures if the number of components ($k$) is larger than the dimension ($n$). But if their means are perturbed, we give a polynomial time algorithm for any $k = \mbox{poly}(n)$ by virtue of our tensor decomposition algorithm. 

Our main technical result is the following:

\begin{theorem}\label{thm:krl}
Let $R\leq  n^{\ell}/2$ for some constant  $\ell \in \N$. Let  $A^{(1)}, A^{(2)}, \dots A^{(\ell)}$ be $n \times R$ matrices with columns of unit norm, and let $\ha{A}^{(1)}, \ha{A}^{(2)}, \dots \ha{A}^{(\ell)} \in \R^{n \times m}$ be their respective $\rho$-perturbations. Then for $\tau = (n/\rho)^{3^{\ell}}$, the Khatri-Rao product satisfies 
\begin{equation}
\kr{\tau}\left(\ha{A}^{(1)} \krp \ha{A}^{(2)}  \krp \dots \krp \ha{A}^{(\ell)} \right) = R  ~~\text{w.p. at least } ~ 1 - \exp\left(- C n^{1/3^{\ell}}\right)
\end{equation}
\end{theorem}

\noindent In general the Kruskal rank {\em adds} \cite{AMR,BCV} but in the framework of smoothed analysis it {\em robustly multiplies}. What is crucial here is that we have a lower bound $\tau$ on how close these vectors are to linearly dependent. In almost all of the applications of tensor methods, we are not given $T$ exactly but rather with some noise. This error could arise, for example, because we are using a finite number of samples to estimate the moments of a distribution. It is the condition number of $\ha{A}^{(1)} \krp \ha{A}^{(1)}  \krp \dots \krp \ha{A}^{(\ell)}$ that will control whether various tensor decomposition algorithms work in the presence of noise. 

Another crucial property our method achieves is exponentially small failure probability for any constant $\ell$, for our polynomial bound on $\tau$.
In particular for $\ell=2$, we show (in Theorem~\ref{thm:kr2}) for $\rho$-perturbations of two $n \times n^2/2$ matrices $U$ and $V$,  the $\kr{\tau} (\Util \krp \Vtil) =n^2/2$ for $\tau=\rho^2/n^{O(1)}$, with probability $1- \exp(-\sqrt{n})$.  We remark that it is fairly straightforward  to obtain the above statement (for $\ell=2$) for failure probability $\delta$,  with $\tau=(n/\delta)^{O(1)}$ (see Remark~\ref{rmk:weakbound} for more on the latter); however, this is not desirable since the running time has a polynomial dependence on the minimum singular value $1/\tau$ (and hence $\delta$). 

We obtain the following main theorem from the above result and from analyzing the stability of the algorithm of Leurgans et al \cite{LRA} (see Theorem~\ref{thm:stability}):

\begin{theorem}\label{thm:main}
Let $R\leq  n^{\lfloor \frac{\ell - 1}{2} \rfloor}/2$ for some constant  $\ell \in \N$. Suppose we are given $\ha{T} + E$ where $\ha{T} $ and $E$ are order $\ell$-tensors and $\ha{T}$ has rank $R$ and is obtained from the above smoothed analysis model. Moreover suppose the entries of $E$ are at most $\eps (\rho/n)^{3^{\ell}}$ where $\eps < 1$. Then there is an algorithm to recover the rank one terms $\otimes_{i=1}^\ell \widetilde{a}^j_i$ up to an additive $\eps$ error. The algorithm runs in time $n^{C 3^{\ell}}$ and succeeds with probability at least  $1 - \exp(- C n^{1/3^{\ell}})$. 
\end{theorem}

As we discussed, tensor methods have had numerous applications in learning. However algorithms that make use of tensor decompositions have traditionally been limited to the full-rank case, and hence can only handle cases when the number of ``components" is at most the dimension. However by using our main theorem above, we can get new algorithms for some of these problems that work even if there are many more ``components" than dimensions. 

\subsubsection*{Multi-view Models (Section~\ref{sec:multi view})} In this setting, each sample is composed of $\ell$ views $\spc{x}{1}, \spc{x}{2}, \dots,\spc{x}{\ell}$ which are conditionally independent given which component $i \in [R]$ the sample is generated from. Hence such a model is specified by $R$ mixing weights $w_i$ and $R$ discrete distributions $\spc{\mu_i}{1},\dots,\spc{\mu_i}{j},\dots,\spc{\mu_i}{\ell}$, one for each view. Such models are very expressive and are used as a common abstraction for a number of inference problems. Anandkumar et al \cite{AHK} gave algorithms in the full rank setting. However, in many practical settings like speech recognition and image classification, the dimension of the feature space is typically much smaller than the  number of components. If we suppose that the distributions that make up the multi-view model are $\rho$-perturbed (analogously to the tensor setting) then we can give the first known algorithms for the overcomplete setting. Suppose that the means $(\spc{\mu_i}{j})$ are $\rho$-perturbed to obtain $\{ \spc{\ha{\mu}_i}{j} \}$. Then:

\begin{theorem}\label{thm:main:multiview}
This is an algorithm to learn the parameters $w_i$ and $\{ \spc{\ha{\mu}_i}{j} \}$ of an $\ell$-view multi-view model with $R \leq n^{\lfloor \frac{\ell - 1}{2} \rfloor} / 2$ components up to an accuracy $\eps$. The running time and sample complexity are at most $\mbox{poly}_\ell(n, 1/\eps, 1/\rho)$ and succeeds with probability at least  $1 - \exp(- C n^{1/3^{\ell}})$ for some constant $C>0$. 
\end{theorem}

\subsubsection*{Mixtures of Axis-Aligned Gaussians (Section~\ref{sec:gaussians})} Here we are given samples from a distribution $F = \sum_{i = 1}^k w_i F_i(\mu_i, \Sigma_i)$ where $F_i(\mu_i, \Sigma_i)$ is a Gaussian with mean $\mu_i$ and covariance $\Sigma_i$ and each $\Sigma_i$ is diagonal. These mixtures are ubiquitous throughout machine learning. Feldman et al \cite{FSO} gave an algorithm for PAC-learning mixtures of axis aligned Gaussians, however the running time is exponential in $k$, the number of components. Hsu and Kakade \cite{HK} gave a polynomial time algorithm for learning mixtures of spherical Gaussians provided that their means are full rank (hence $k \leq n$). Again, we turn to the framework of smoothed analysis and suppose that the means are $\rho$-perturbed. In this framework, we can give a polynomial time algorithm for learning mixtures of axis-aligned Gaussians for any $k = \mbox{poly}(n)$. Suppose that the means of a mixture of  axis-aligned Gaussians and suppose the means have been $\rho$-perturbed to obtain $\ha{\mu}_i$. Then

\begin{theorem}\label{thm:main:gaussians}
There is an algorithm to learn the parameters $w_i$, $\ha{\mu}_i$ and $\Sigma_i$ of a mixture of $k \leq n^{\lfloor \frac{\ell - 1}{2} \rfloor}/(2 \ell)$ axis-aligned Gaussians up to an accuracy $\eps$. The running time and sample complexity are at most $\mbox{poly}_\ell(n, 1/\eps, 1/\rho)$ and succeeds with probability at least  $1 - \exp(- C n^{1/3^{\ell}})$ for some constant $C>0$. 
\end{theorem}

We believe that our new algorithms for overcomplete tensor decomposition will have further applications in learning. Additionally this framework of studying distribution learning when the parameters of the distribution we would like to learn are not chosen adversarially, seems quite appealing. 

\begin{remark}\label{remark:generic}
Recall, our main technical result is that the Kruskal rank {\em robustly} multiplies. In fact, is is easy to see that for a generic set of vectors it multiplies \cite{AMR}. This observation, in conjunction with the algorithm of Leurgans et al \cite{LRA} yields an algorithm for tensor decomposition in the overcomplete case. Another approach to overcomplete tensor decomposition was given by \cite{DCC} which works up to $r \leq n^{\lfloor \frac{\ell}{2} \rfloor}$. However these algorithms assume that we know $T$ {\em exactly}, and are not known to be stable when we are given $T$ with noise. The main issue is that these algorithms are based on solving a linear system which is full rank if the factors of $T$ are generic, but what controls whether or not these linear systems can handle noise is their condition number. 
\end{remark}

\noindent Alternatively, algorithms for overcomplete tensor decomposition that assume we know $T$ exactly would not have any applications in learning because we would need to take too many samples to have a good enough estimate of $T$ (i.e. the low-order moments of the distribution). 

%In recent work, Goyal et al \cite{GVX} gave an algorithm for tensor decompositions that is similar to the one in Corollary~\ref{corr:stability}. They apply this to independent component analysis and learning mixtures of spherical Gaussians; however their work assumes that the parameters $M$ (e.g.the factors in a tensor decomposition, or the means of a mixture of Gaussians,)
%satisfy $\kr{\tau}(M^{\krp \ell}) =R$ for some constant $\ell$. 
%We note that the technical core of our paper is to establish that this indeed happens in the smoothed analysis setting with high probability.

%In recent work, Goyal et al \cite{GVX} gave an algorithm for overcomplete independent component analysis that makes use of overcomplete tensor decompositions. We note that Goyal et al \cite{GVX} assume that there is some cumulant that is bounded away from zero. This has a similar spirit to the statement of Corollary~\ref{corr:stability} where we show that it suffices to have bounds on the condition numbers of certain matrices to get a stable algorithm for tensor decomposition. However here the technical core of our paper is in establishing that this indeed happens with high probability in the smoothed analysis setting.

In recent work, Goyal et al \cite{GVX} also made use of robust algorithms for overcomplete tensor decomposition, and their main application is underdetermined independent component analysis (ICA). The condition that they need to impose on the tensor holds generically (like ours, see e.g. Corollary~\ref{corr:stability}) and can show in a smoothed analysis model that this condition holds with inverse polynomial failure probability. However here our focus was on showing a lower bound for the condition number of $M^{\krp \ell}$ that does not depend (polynomially) on the failure probability.  We focus on the failure probability being small (in particular, exponentially small), because in smoothed analysis, the perturbation is ``one-shot'' and if it does not result in an easy instance, you cannot ask for a new one!

%gave an algorithm for overcomplete independent component analysis (ICA) that makes use of overcomplete tensor decompositions. We note that Goyal et al \cite{GVX} assume that 
%there is some cumulant
%$M^{\krp \ell}$ is well-conditioned for some constant $\ell$, where  $M$ is the matrix of parameters (e.g.the factors in a tensor decomposition, or means of a mixture of Gaussians). 
%satisfy $M^{\krp \ell}$ is bounded away from $0$ for some constant $\ell$
%satisfy $\kr{\tau}(M^{\krp \ell}) =R$ for some constant $\ell$, %or $\sigma_R(M^{\krp \ell})$ that the parameters $M$ (e.g.the factors in a tensor decomposition) satisfy  for some constant $\ell$ 
%that is bounded away from zero. 
%This has a similar spirit to the statement of Corollary~\ref{corr:stability} where we show that it suffices to have bounds on the condition numbers of certain matrices to get a stable algorithm for tensor decomposition. However here the technical core of our paper is in establishing that this indeed happens with high probability in the smoothed analysis setting.

\subsection{Our Approach}

Here we give some intuition for how we prove our main technical theorem, at least in the $\ell = 2$ case. Recall, we are given two matrices $U^{(1)}$ and $U^{(2)}$ whose $R$ columns are $\rho$-perturbed to obtain $\ha{U}^{(1)}$ and $\ha{U}^{(2)}$ respectively. Our goal is to prove that if $R \leq \frac{n^2}{2}$ then the matrix $\ha{U}^{(1)} \odot \ha{U}^{(2)}$ has smallest singular value that is at least $\mbox{poly}(1/n, \rho)$ with high probability. In fact, it will be easier to work with what we call the {\em leave-one-out distance} (see Definition~\ref{def:leaveoneout}) as a surrogate for the smallest singular value (see Lemma~\ref{lem:loo}). Alternatively, if we let $x$ and $y$ be the first columns of $\ha{U}^{(1)}$ and $\ha{U}^{(2)}$ respectively, and we set $$\calU = \mbox{span}(\{\ha{U}^{(1)}_i \otimes \ha{U}^{(2)}_i, 2 \leq i \leq R\})$$ then we would like to prove that with high probability $x \otimes y$ has a non-negligible projection on the orthogonal complement of $\calU$. This is the core of our approach. Set $\calV$ to be the orthogonal complement of $\calU$. In fact, we prove that for {\em any} dimension at least $\frac{n^2}{2}$ subspace $\calV$, with high probability $x \otimes y$ has a non-negligible projection onto $\calV$. 

How can we reason about the projection of $x \otimes y$ onto an arbitrary (but large) dimensional subspace? If $\calV$ were (say) the set of all low-rank matrices, then this would be straightforward. But what complicates this is that we are looking at the projection of a rank one matrix onto a large dimensional subspace of matrices, and these two spaces can be structured quite differently. A natural approach is to construct matrices $M_1, M_2, ... , M_p \in \calV$ so that with high probability at least one quadratic form $x^T M_i y$ is non-negligible. Suppose the following condition were met (in which case we would be done): Suppose that there is a large set $S$ of indices so that each vector $x^T M_i$ has a large projection onto the orthogonal complement of $\mbox{span}(\{x^T M_i, i \in S\})$. In fact, if such a set $S$ exists with high probability then this would yield our main technical theorem in the $\ell = 2$ case. Our main step is in constructing a family of matrices $M_1, M_2, ... M_p$ that help us show that $S$ is large. We call this an $(\theta, \delta)$-orthogonal system (see Definition~\ref{def:system}). The intuition behind this definition is that if we reveal a column in one of the $M_i$'s that has a significant orthogonal component to all of the columns that we have revealed so far, this is in effect a fresh source of randomness that can help us add another index to the set $S$. See Section~\ref{sec:multiply} for a more complete description of our approach in the $\ell = 2$ case. The approach for $\ell > 2$ relies on the same basic strategy but requires a more delicate induction argument. See Section~\ref{sec:higherproducts}.

\section{Prior Algorithms} \label{sec:review}

Here we review the algorithm of Leurgans et al \cite{LRA}. It has been discovered many times in different settings. It is sometimes referred to as ``simultaneous diagonalization'' or as Chang's lemma \cite{C}. 

Suppose we are given a third-order tensor $T = \sum_{i = 1}^R u_i \otimes v_i \otimes w_i$ which is $n \times m \times p$. Let $U, V$ and $W$ be matrices whose columns are $u_i$, $v_i$ and $w_i$ respectively. Suppose further that (1) $rank(U) = rank(V) = R$ and (2) $\mbox{k-rank}(W) \geq 2$. Then we can efficiently recover the factors of $T$. %using the algorithm {\sc Decompose}:

\newcommand{\utilde}{\widetilde{u}}
We present the algorithm {\sc Decompose} and its analysis assuming $n=m=R$. Any instance with $rank(U) = rank(V) = R$ can be reduced to this case as follows: find the span of the vectors $\{ \utilde_{j,k} \}$, where $\utilde_{j,k}$ is the $n$ dimensional vector whose $i$th entry is $T_{ijk}$. This span must be precisely the span of the columns of $U$.\footnote{It is easy to see that the span is contained in the span of the columns of $U$. To see equality, we observe that if the span is $R-1$ dimensional, then projecting each of the $u_i$s on to the span gives a {\em different} decomposition, and this contradicts Kruskal's uniqueness theorem, which holds in this case.} Thus we can pick some orthonormal basis for this span, and write $T$ as an $R \times m \times p$ tensor. We can perform this operation again (along the second mode) to move to an $R \times R \times p$ tensor.

\begin{theorem} \cite{LRA}, \cite{C} \label{thm:decompose}
Given a tensor $T$ there exists an algorithm that runs in polynomial time and recovers the (unique) factors of $T$ provided that (1) $rank(U) = rank(V) = R$ and (2) $\mbox{k-rank}(W) \geq 2$.
\end{theorem}

\begin{proof}
The algorithm is to pre-process as above (i.e., obtain $m=n=R$), and then run {\sc Decompose} stated below. Let us thus analyze {\sc Decompose} with $m,n$ being $R$.

We can write $T_a = U D_a V^T$ where $D_a = \mbox{diag}(a^T w_1, a^T w_2, ..., a^T w_n)$ and similarly $T_b = U D_b V^T$ where
$D_b = \mbox{diag}(b^T w_1, b^T w_2, ..., b^T w_n).$
Moreover we can write $T_a (T_b)^{-1} = U D_a D_b^{-1} U^{-1}$ and $T_a^T (T_b^T)^{-1}= V D_a D_b^{-1} V^{-1}$. So we conclude $U$ and $V$ diagonalize $T_a (T_b)^{-1} $ and $((T_b)^{-1} T_a)^T $ respectively. 
Note that almost surely the diagonals entries of $D_a D_b^{-1}$ are distinct (Claim~\ref{claim:separation-evals}). Hence the eigendecompositions of $T_a (T_b)^{-1} $ and $(T_b)^{-1} (T_a) $ are unique, and we can pair up columns in $U$ and columns in $V$ based on their eigenvalues (we pair up $u$ and $v$ if their eigenvalues are equal). 
We can then solve a linear system to find the remaining factors (columns in $W$) and since this is a valid decomposition, we can conclude that these are also the true factors of $T$ appealing to Kruskal's uniqueness theorem \cite{Kru}. 
\end{proof}

\noindent In fact, this algorithm is also stable, as Goyal et al \cite{GVX} also recently showed. It is intuitive that if $U$ and $V$ are well-conditioned and each pair of columns in $W$ is well-conditioned then this algorithm can tolerate some inverse polynomial amount of noise. For completeness, we give a robustness analysis of {\sc Decompose} in Appendix~\ref{asec:stable}. %and still recover factors that are $\epsilon$ close to the true factors. 
%\vnote{Is it exactly this algorithm that they analyze?}
\begin{condition}\label{condition:approx}
\begin{enumerate}
\item The condition numbers $\kappa(U), \kappa(V) \leq \kappa$,
\item The column vectors of $W$ are not close to parallel: for all $i \neq j$, $\|\frac{w_i}{\|w_i\|} - \frac{w_j}{\|w_j\|} \|_2 \geq \delta$ ,
\item The decompositions are bounded : for all $i$, $\|u_i\|_2, \|v_i\|_2, \|w_i\|_2 \leq C$.
\end{enumerate}
\end{condition}

\begin{theorem} \label{thm:stability}
Suppose we are given tensor $T + E \in \R^{m \times n \times p}$ with the entries of $E$ being bounded by $ \epsilon\cdot  \mbox{poly}(1/\kappa, 1/n, 1/\delta)$ and moreover $T$ has a decomposition 
$T = \sum_{i = 1}^R u_i \otimes v_i \otimes w_i$ that
satisfies Condition~\ref{condition:approx}. Then there exists an efficient algorithm that returns {\em each rank one term} in the decomposition of $T$  (up to renaming), within an additive error of $\epsilon$. 
\end{theorem}
\noindent As before, the algorithm is to preprocess so as to obtain $m=n=R$, and then run {\sc Decompose}. The preprocessing step is slightly different because of the presence of error -- instead of considering the span of the $\{ \utilde_{j,k} \}$ as above, we need to look at the span of the top $R$ singular vectors of the matrix whose columns are $\utilde_{j,k}$. If $\norm{E}_F$ is small enough (in terms of $\kappa, \delta, n$), the span of these top singular vectors suffices to obtain an approximation to the vectors $u_i$ (see Appendix~\ref{asec:stable}).

Note that the algorithm is limited by the condition that $rank(U) = rank(V) = R$ since this requires that $R \leq \min(m, n)$. But as we have seen before, by ``flattening" a higher order tensor,  we can handle overcomplete tensors. The following is an immediately corollary of Theorem~\ref{thm:stability}:
%\vnote{Do we need this statement? Or is it fine if it is just in words?}
\begin{corollary} \label{corr:stability}
Suppose we are given an order-$\ell$ tensor $T + E \in \R^{ n^{\times \ell}}$ with the entries of $E$ being bounded by $ \epsilon\cdot  \mbox{poly}_\ell(1/\kappa, 1/n, 1/\delta)$, and matrices $\spc{U}{1},\spc{U}{2} \dots \spc{U}{\ell} \in \R^{ n \times r}$, whose columns give a rank-$r$ decomposition %of $T$: 
$T = \sum_{i = 1}^R u^{(1)}_i \otimes u^{(2)}_i \otimes \dots \otimes u^{(\ell)}_i$. If Condition~\ref{condition:approx} is satisfied by 
$$U=\spc{U}{1} \krp \spc{U}{2} \krp \dots \krp \spc{U}{\lfloor \frac{\ell-1}{2} \rfloor}~, ~ V=\spc{U}{\lfloor \frac{\ell-1}{2} \rfloor+1 } \krp  \dots \krp \spc{U}{2 \lfloor \frac{\ell-1}{2} \rfloor} ~\text{ and }~ W=\begin{cases} \spc{U}{\ell} & \text{if } \ell \text{ is odd}\\ U^{(\ell-1)} \krp U^{(\ell)}& \text{ otherwise}  \end{cases}$$ then there exists an efficient algorithm that computes each rank one term in this decomposition up to an additive error of $\epsilon$. 
\end{corollary}
\noindent Note that Corollary~\ref{corr:stability} does not require the decomposition to be symmetric. Further, any tri-partition of the $\ell$ modes that satisfies Condition~\ref{condition:approx} would have sufficed. To understand how large a rank we can handle, the key question is: {\em When does the Kruskal rank (or rank) of $\ell$-wise Khatri-Rao product become $R$?} 

\begin{fragment*}[t]
\caption{
\label{alg:decomp}{\sc Decompose}, \textbf{Input: } $T \in \R^{R \times R \times R}$ \vspace*{0.01in}
}

\begin{enumerate} \itemsep 0pt
\small 
\item Let $T_a = T(\cdot, \cdot, a), T_b = T(\cdot, \cdot, b)$ where $a, b$ are uniformly random unit vectors in $\Re^p$
\item Set $U$ to be the eigenvectors of $T_a (T_b)^{-1}$ % = U D_a U^{-1}$
\item Set $V$ to be the eigenvectors of $\left( (T_b)^{-1} T_a\right)^T$ % = V D_b V^{-1}$
\item Solve the linear system $T = \sum_{i = 1}^n u_i \otimes v_i \otimes w_i$ for the vectors $w_i$
\item Output $U, V, W$
\end{enumerate} 

\end{fragment*}

The following lemma is well-known (see  \cite{BCV} for a robust analogue) and is known to be tight in the worst case. This allows us to handle a rank of $R \approx \ell n/2$.
\begin{lemma}
$\kr{}(U \odot V) \geq \min \big(\kr{}(U) + \kr{}(V)-1 , R \big)$
\end{lemma}
But, for generic vectors set of vectors $U$ and $V$, a much stronger statement is true \cite{AMR}: $\kr{}(U \odot V) \geq \min\big( \kr{}(U) \times \kr{}(V), R \big)$.
%\begin{lemma} \cite{AMR} 
%For a generic set of vectors $U$ and $V$, $\kr{}(U \odot V) \geq \min\big( \kr{}(U) \times \kr{}(V), r \big)$
%\end{lemma}
%\vnote{Is this generic kruskal-rank growing multiplicatively known to be true as stated?}
Hence given a generic order $\ell$ tensor $T$ with $R \leq n^{\lfloor (\ell - 1)/2 \rfloor}$, ``flattening" it to order three and appealing to Theorem~\ref{thm:decompose} finds the factors uniquely. The algorithm of \cite{DCC} follows a similar but more involved approach, and works for $R \leq n^{\lfloor (\ell )/2 \rfloor}$. 

%\begin{center} 
%{\em We will refer to this as the ``Khatri-Rao Trick"!}
% \end{center}
% 
However in learning applications we are not given $T$ exactly but rather an approximation to it. Our goal is to show that the Kruskal rank {\em robustly} multiplies typically, so that these types of tensor algorithms will not only work in the exact case, but are also necessarily stable when we are given $T$ with some noise. 
In the next section, we show that in the smoothed analysis model, the robust Kruskal rank multiplies on taking Khatri-Rao products. This then establishes our main result Theorem~\ref{thm:main},  assuming Theorem~\ref{thm:krl} which we prove in the next section.

\begin{proof}[Proof of Theorem~\ref{thm:main}]
As in Corollary~\ref{corr:stability}, let $U=\spc{\Util}{1} \krp  \dots \krp \spc{\Util}{\lfloor \frac{\ell-1}{2} \rfloor}~,~  V=\spc{\Util}{\lfloor \frac{\ell-1}{2} \rfloor+1 } \krp  \dots \krp \spc{\Util}{\ell-1}$  and $W=\spc{\Util}{\ell}$. 
Theorem~\ref{thm:krl} shows that with probability $1-\exp\big(-n^{1/3^{O(\ell)}}\big)$ over the random $\rho$-perturbations, $\kappa_R(U),\kappa_R(V) \le  (n/\rho)^{3^{\ell}}$. Further, the columns $W$ are $\delta=\rho/n$ far from parallel with high probability.  Hence, Corollary~\ref{corr:stability} implies Theorem~\ref{thm:main}.  
\end{proof}

\section{The Khatri-Rao Product Robustly Multiplies}\label{sec:multiply}

In the exact case, it is enough to show that the Kruskal rank almost surely multiplies and this yields algorithms for overcomplete tensor decomposition if we are given $T$ {\em exactly} (see Remark~\ref{remark:generic}). But if we want to prove that these algorithms are stable, we need to establish that even the robust Kruskal rank (possibly with a different threshold $\tau$) also multiplies. This ends up being a very natural question in {\em random matrix theory}, albeit the Khatri-Rao product of two perturbed vectors in $\R^n$ is far from a perturbed vector in $\R^{n^2}$. 

Formally, suppose we have two matrices $U$ and $V$ with columns $u_1, u_2, \dots, u_{R} $ and $v_1, v_2, \dots, v_{R}$ in $\R^n$. Let $\ha{U}, \ha{V}$ be $\rho$-perturbations of $U,V$ i.e. for each $i \in [R]$, we perturb $u_i$ with an (independent) random gaussian perturbation of norm $\rho$ to obtain $\ha{u}_i$ (and similarly for $\ha{v}_i$). Then we show the following:
%\TODO{Reword this statement for higher-order tensors.}
%vectors $\set{u_i}_{i \in [m]}$ and $\set{v_i}_{i \in [m]}$ are of unit norm, and the independent gaussian perturbations to each of these vectors have norm $\rho$. Then for $\tau = n^4/\rho^2$, the Khatri-Rao product satisfies $\kr{\tau}(\widehat{U} \odot \widehat{V}) = n^2/2$ with probability at least $1 - \exp(-\sqrt{n})$.
\begin{theorem}\label{thm:kr2}
Suppose $U,V$ are $n \times R$ matrices and let $\Util, \Vtil$ be $\rho$-perturbations of $U,V$ respectively.  Then for any constant $\delta \in (0,1)$, $R \le \delta n^2$ and $\tau = n^{O(1)}/\rho^2$, the Khatri-Rao product satisfies $\kr{\tau}(\Util \krp \Vtil) = R$ with probability at least $1 - \exp(-\sqrt{n})$.
\end{theorem}
\begin{remark}
The natural generalization where the vectors $u_i$ and $v_i$ are in different dimensional spaces also holds. We omit the details here.
\end{remark}

%\noindent The above result allows us to (stably) decompose order five tensors even if the rank is as large as $\delta n^2$, in the smoothed analysis framework. 

In general, a similar result holds for $\ell$-wise Khatri-Rao products which allows us to handle rank as large as $\delta n^{\lfloor \frac{\ell -1}{2} \rfloor}$ for $\ell = O(1)$. Note that this does not follow by repeatedly applying the above theorem (say applying the theorem to $U \odot V$ and then taking $\odot W$), because perturbing the entries of $(U \odot V)$ is not the same as $\Util \odot \Vtil$. In particular, we have only $\ell \cdot nR$ ``truly'' random bits, which are the perturbations of the columns of the base matrices.  The overall structure of the proof is the same, but we need additional ideas followed by a delicate induction.
\begin{theorem}\label{thm:krl}
For any $\delta \in (0,1)$, let $R=\delta n^{\ell}$ for some constant  $\ell \in \N$. Let  $U^{(1)}, U^{(2)}, \dots U^{(\ell)}$ be $n \times R$ matrices with unit column norm, and let $\ha{U}^{(1)}, \ha{U}^{(2)}, \dots \ha{U}^{(\ell)} \in \R^{n \times m}$ be their respective $\rho$-perturbations. Then for $\tau = (n/\rho)^{3^{\ell}}$, the Khatri-Rao product satisfies 
\begin{equation}
\kr{\tau}\left(\ha{U}^{(1)} \krp \ha{U}^{(1)}  \krp \dots \krp \ha{U}^{(\ell)} \right) = n^\ell/2  ~~\text{w.p. at least } ~ 1 - \exp\left(- \delta n^{1/3^{\ell}}\right)
\end{equation}
\end{theorem}

Let %us first focus on the robust K-rank lower bound for two-wise products (Theorem~\ref{thm:kr2}), and let 
$A$ denote the $n^\ell \times R$ matrix $\Util^{(1)} \krp \Util^{(2)} \krp \dots \krp \Util^{(\ell)}$ for convenience. The theorem states that the smallest singular value of $A$ is lower-bounded by $\tau$. %(as opposed to an upper bound on the largest singular value, that is much easier). 
%\begin{question}

\emph{How can we lower bound the smallest singular value of $A$?}
%\end{question}
We define a quantity which is can be used as a proxy for the least singular value and is simpler to analyze.
\begin{definition}\label{def:leaveoneout}
For any matrix $A$ with columns $A_1, A_2, \dots A_R$,  the leave-one-out distance is
\[ \ell(A) = \min_i ~ dist(A_i, span\{A_j\}_{j \neq i}).\]
\end{definition}

The leave-one-out distance is a good proxy for the least singular value, if we are not particular about losing multiplicative factors that are polynomial in size of the matrix.
\begin{lemma}\label{lem:loo}
For any matrix $A$ with columns $A_1, A_2, \dots A_R$, we have 
$\frac{\ell(A)}{\sqrt{R}} \leq \sigma_{min}(A) \leq \ell(A)$.
\end{lemma}

We will show that each of the vectors $A_{i} = \ha{u}^{(1)}_{i} \otimes \ha{u}^{(2)}_i \otimes\dots\otimes \ha{u}^{(\ell)}_{i} $ has a reasonable projection (at least $n^{\ell/2}/\tau$) on the space orthogonal to the span of the rest of the vectors $\spn\left(\set{A_j : j \in [R]-\set{i} }\right)$ with high probability. We do not have a good handle on the space spanned by the rest of the $R-1$ vectors, so we will prove a more general statement in Theorem~\ref{thm:kr}: % which holds for any dimension $R-1$ subspace:
we will prove that a perturbed vector $\ha{x}^{(1)} \otimes \dots \otimes \ha{x}^{(\ell)}$ has a reasonable projection onto {\em any} (fixed) subspace $\calV$ w.h.p., as long as dim$(\calV)$ is $\Omega(n^\ell)$. To say that a vector $w$ has a reasonable projection onto $\calV$, we just need to exhibit a set of vectors in $\calV$ such that one of them have a large inner product with $w$. This will imply our the required bound on the singular value of $A$ as follows: 
\begin{enumerate}
\item Fix an $i \in [R]$ and apply Theorem~\ref{thm:kr} with $x^{(t)}=u^{(t)}_i$ for all $t \in [\ell]$,  and $\calV$ being the space orthogonal to rest of the vectors $A_j$. 
\item Apply a union bound over all the $R$ choices for $i$.
\end{enumerate}
%\subsection{Projections of Tensor products of }
%
%The following theorem implies the required statement for the multiplicativity of the Kruskal rank in the smooth setting. 
We now state the main technical theorem about projections of perturbed product vectors onto arbitrary subspaces of large dimension. 
\begin{theorem}\label{thm:kr}
For any constant $\delta \in (0,1)$, given any subspace $\calV$ of dimension $\delta \cdot n^{\ell}$ in $\R^{n^{\times \ell}}$, there exists tensors $T_1, T_2, \dots T_r$ in $\calV$ of unit norm ($\norm{\cdot}_F = 1$), such that for random $\rho$-perturbations $\xtil^{(1)},\xtil^{(2)},\dots,\xtil^{(\ell)} \in \R^n$ of any vectors $x^{(1)},x^{(2)},\dots,x^{(\ell)} \in \R^n$, we have   
\begin{equation} \label{eq:thm:kr}
\Pr\left[ \exists j \in [r] ~\text{s.t} ~ \norm{T_j\left(\xtil^{(1)},\xtil^{(2)},\dots,\xtil^{(\ell)} \right)} \ge \rho^\ell \left(\frac{1}{n}\right)^{3^\ell} \right] \ge 1- \exp\left(- \delta n^{1/(2\ell)^{\ell}} \right)
\end{equation}
\end{theorem}
%\vnote{Please see the remark, and reword if necessary.}
\begin{remark}\label{rmk:weakbound}
Since the squared length of the projection is a degree $2\ell$ polynomial of the (Gaussian) variables $x_i$, we can apply standard anti-concentration results (Carbery-Wright, for instance) to conclude that the smallest singular value (in Theorem~\ref{thm:kr}) is at least an inverse polynomial, with failure probability at most an inverse polynomial. This approach can only give a singular value lower bound of $\text{poly}_\ell (p/n)$ for a failure probability of $p$, which is not desirable since the running time depends on the smallest singular value.%much weaker than Theorem~\ref{thm:kr}.
\end{remark}

\begin{remark}\label{rem:baddelta}
For meaningful guarantees, we will think of $\delta$ as a small constant or $n^{-o(1)}$ (note the dependence of the error probability on $\delta$ in eq~\eqref{eq:thm:kr}). For instance, as we will see in section~\ref{sec:higherproducts}, we can not hope for exponential small failure probability when $\calV \subseteq \R^{n^2}$ has dimension $n$. 
\end{remark}
%\vnote{On second thoughts, do we really need a ``proof'' of this in the appendix? It's already outlines above.}
%We will defer the details of the proof of Theorem~\ref{thm:kr} assuming Theorem~\ref{thm:krl} to the appendix~\ref{app:thmproofs}.  
The following restatement of Theorem~\ref{thm:kr} gives  a sufficient condition about the singular values of a matrix $P$ of size $r \times n^{\ell}$, that gives a strong anti-concentration property for values attained by vectors obtained by the tensor product of perturbed vectors.  This alternate view of Theorem~\ref{thm:kr} will be crucial in the inductive proof for higher $\ell$-wise products in section~\ref{sec:higherproducts}. 
   
\begin{theorem}[Restatement of Theorem~\ref{thm:kr}]\label{thm:restatement}
Given any constant $\delta_\ell \in (0,1)$ and any matrix $T$ of size $r \times (n^{\ell})$ such that $\sigma_{\delta n^{\ell}} \ge \eta$, then for random $\rho$-perturbations $\xtil^{(1)},\xtil^{(2)},\dots,\xtil^{(\ell)} \in \R^n$ of any vectors $x^{(1)},x^{(2)},\dots,x^{(\ell)} \in \R^n$, we have   
\begin{equation}
\Pr\left[ \norm{M \left( \xtil^{(1)},\xtil^{(2)},\dots,\xtil^{(\ell)} \right)} \ge \eta \rho^\ell \left(\frac{1}{n}\right)^{3^{O(\ell)}} \right] \ge 1- \exp\left(-\delta n^{1/3^{\ell}} \right)
\end{equation}
\end{theorem}
%\TODO{Check the dependency on $\delta$.}
\begin{remark}
Theorem~\ref{thm:kr} follows from the above theorem by choosing an orthonormal basis for $\calV$ as the rows of $T$. The other direction follows by choosing $\calV$ as the span of the top $\delta_\ell n^{\ell}$ right singular vectors of $T$. 
\end{remark}

\begin{remark}
Before proceeding, we remark that both forms of Theorem~\ref{thm:kr} could be of independent interest.  For instance, it follows from the above (by a small trick involving partitioning the coordinates), that a vector $\xtil^{\otimes \ell}$ has a non-negligible projection into any $cn^\ell$ dimensional subspace of $\R^{n^\ell}$ with probability $1- \exp(-f_\ell(n))$. For a vector $x \in \R^{n^\ell}$ whose entries are all independent Gaussians, such a claim follows easily, with probability roughly $1-\exp(-n^\ell)$.  The key difference for us is that $\xtil^{\otimes \ell}$ has essentially just $n$ bits of randomness, so many of the entries are highly correlated. So the theorem says that even such a correlated perturbation has enough mass in any large enough subspace, with high enough probability.  A natural conjecture is that the probability bound can be improved to $1-\exp(-\Omega(n))$, but it is beyond the reach of our methods.
%Another way of seeing our result is as follows: suppose we have some {\em truly random} bits $x_1, x_2, \dots, x_n$, and we consider an $m$ dimensional vector whose entries degree $\ell$ polynomials of the $x_i$, and $m > n$. It is natural to ask if this vector has random-like properties in $\R^m$ (having a good projection into any given large subspace, say). The theorem proves this when $m = n^\ell$ and the polynomials are all the degree $\ell$ monomials in $\{ x_i \}$.
\end{remark}

\subsection{Khatri-Rao Product of Two Matrices}

We first show Theorem~\ref{thm:restatement} for the case $\ell=2$. This illustrates the main ideas underlying the general proof. % The following proposition is equivalent to setting $\ell=2$ in Theorem~\ref{thm:restatement}.%The proof of Proposition~\ref{prop:matrixcase} will be very instructive, since we will use the same ideas to carry out the proof for the general case of $\ell$-wise products.
\begin{prop}\label{prop:matrixcase}
Let $0 < \delta < 1$ and $M$ be a $\delta n^2 \times n^2$ matrix with $\sigma_{\delta n^2} (M) \ge \tau$.  Then for random $\rho$-perturbations $\xtil, \ytil$ of any two $x, y \in \R^n$, we have   
\begin{equation}
\Pr\left[ \norm{M \left( \xtil \otimes \ytil \right)} \ge \frac{\tau \rho}{n^{O(1)}} \right] \ge 1- \exp\left(-\sqrt{\delta  n} \right).
\end{equation}
\end{prop}

The high level outline is now the following. Let $\calU$ denote the span of the top $\delta n^2$ singular vectors of $M$.  We show that for $r = \Omega(\sqrt{n})$, there exist $n \times n$ matrices $M_1, M_2, \dots, M_r$ whose columns satisfy certain orthogonal properties we define, and additionally $\vect(M_i) \in \calU$ for all $i \in [r]$.  We use the orthogonality properties to show that $(\xtil \otimes \ytil)$ has an $\rho/\poly(n)$ dot-product with at least one of the $M_i$ with probability $\ge 1-\exp(-r)$.

\paragraph{The $\theta$-orthogonality property.}
In order to motivate this, let us consider some matrix $M_i \in \R^{n\times n}$ and consider $M_i (x \otimes y)$. This is precisely $y^T M_i x$.  Now suppose we have $r$ matrices $M_1, M_2, \dots, M_r$, and we consider the sum $\sum_i (y^T M_i x)^2$.  This is also equal to $\norm{Q(y) x}^2$, where $Q(y)$ is an $r \times n$ matrix whose $(i,j)$th entry is $\iprod{y, (M_i)_j}$ (here $(M_i)_j$ refers to the $j$th column in $M_i$).

Now consider some matrices $M_i$, and suppose we knew that $Q(\ytil)$ has $\Omega(r)$ singular values of magnitude $\ge 1/n^2$.  Then, an $\rho$-perturbed vector $\xtil$ has at least $\rho/n$ of its norm in the space spanned by the corresponding right singular vectors, with probability $\ge 1-\exp(-r)$ (Fact~\ref{lem:sv:anticonc}).  Thus we get 
\[ \Pr[ \norm{Q(\ytil) \xtil} \ge \rho/n^3 ] \ge 1-\exp(-r). \]
So the key is to prove that the matrix $Q(\ytil)$ has a large number of ``non-negligible'' singular values with high probability (over the perturbation in $\ytil$). For this, let us examine the entries of $Q(\ytil)$.  For a moment suppose that $\ytil$ is a gaussian random vector $\sim \calN(0, \rho^2 I)$ (instead of a perturbation).  Then the $(i,j)$th entry of $Q(\ytil)$ is precisely $\iprod{\ytil, (M_i)_j}$, which is distributed like a one dimensional gaussian of variance $\rho^2 \norm{(M_i)_j}^2$.  If the entries for different $i,j$ were independent, standard results from random matrix theory would imply that $Q(\ytil)$ has many non-negligible singular values.

However, this could be far from the truth.  Consider, for instance, two vectors $(M_i)_j$ and $(M_{i'})_{j'}$ that are parallel.  Then their dot products with $\ytil$ are highly correlated.  However we note, that as long as $(M_{i'})_{j'}$ has a reasonable component orthogonal to $(M_i)_j$, the distribution of the $(i,j)$ and $(i',j')$th entries are ``somewhat'' independent. We will prove that we can roughly achieve such a situation.  This motivates the following definition.

\begin{definition}\label{def:system}[Ordered $\theta$-orthogonality]
A sequence of vectors $v_1, v_2, \dots, v_n$ has the ordered $\theta$-orthogonality property if for all $1\le i\le n$, $v_i$ has a component of length $\ge \theta$ orthogonal to $\spn\{ v_1, v_2, \dots, v_{i-1} \}$.  
\end{definition}

Now we define a similar notion for a sequence of matrices $M_1, M_2, \dots, M_r$, which says that a large enough subset of columns should have a certain $\theta$-orthogonality property.  %It says that we need to have a set of $\delta m$ indices $I = \{i_1, i_2, \dots, i_{\delta m}\}$, such that vectors given by the columns $I$ of $M_j$ satisfy the ordered $\theta$-orthogonal property.
More formally,
\begin{definition}[Ordered $(\theta,\delta)$-orthogonal system]\label{def:toos}
A set of $n \times m$ matrices $M_1, M_2, \dots, M_r$ form an {\em ordered $(\theta,\delta)$-orthogonal system} if there exists a permutation $\pi$ on $[m]$ such that the first $\delta m$ columns satisfy the followng property: for $i\le \delta m$ and every $j \in [R]$, the $\pi(i)$th column of $M_j$  has a projection of length $\ge \theta$ orthogonal to the span of all the vectors given by the columns $\pi(1), \pi(2), \dots, \pi(i-1),\pi(i)$ of all the matrices $M_1, M_2, \dots M_r$ other than itself (i.e. the $\pi(i)$th column of $M_j$). %, as well as the $\pi(i)$th column of $Q_{j'}$ for all $j' \in [r]-\set{j}$.
%\begin{equation}
%\Pr_{i\in [m]}\left[\forall j \in[r], ~ \norm{\Pi^{\perp}_{j} M_j\left(\pi(i)\right)} \ge \theta \cdot \norm{M_j\left(\pi(i)\right)} \right] \ge \delta ~\text{where}
%\end{equation}
%where $\Pi^{\perp}_{j}$ is the projection matrix $\perp^{r}$ to $\text{span}\left(\set{ M_{j'}(\pi(i')): i' \in [i-1], j' \in [r]} \cup \set{ M_{j'}(i): j'\in [r]-{j}}\right)$. 
\end{definition}
%\TODO{Slightly stronger definition than what we intuitively described.}

The following lemma shows the use of an ordered $(\theta,\delta)$ orthogonal system: a matrix $Q(\ytil)$ constructed as above starting with these $M_i$ has many non-negligible singular values with high probability.

\begin{lemma}[Ordered $\theta$-orthogonality and perturbed combinations.]\label{lem:toosystem:implies}
Let $M_1, M_2, \dots, M_r$ be a set of $n \times m$ matrices of bounded norm ($\norm{\cdot}_F \ge 1$) that are $(\theta,\delta)$ orthogonal for some parameters $\theta, \delta$, and suppose $r\le \delta m$. Let $\xtil$ be an $\rho$-perturbation of $x \in \R^{n}$. Then the $r \times m$ matrix $Q(\xtil)$ formed with the $j$th row of $\left(Q(\xtil)\right)_j$ being $\transpose{\xtil} M_j$ satisfies %(for $\scons=\Omega(1)$?)
$$ \Pr_{x} \left[ \sigma_{r/2}\left( Q(\xtil) \right) \ge \frac{\rho \theta}{n ^4} \right] \ge 1- \exp\left(-r \right)$$
\end{lemma}

We defer the proof of this Lemma to section~\ref{sec:theta-suffice}.
%
%\begin{lemma}\label{lem:theta-suffices} \TODO{statement later is much nicer!}
%Let $N_1, N_2, \dots, N_r$ be $n \times r$ matrices such that $v_{(i,j)}$, defined to be the $j$th column of $N_i$, satisfy the ordered $\theta$-orthogonality property when ordered lexicographically by $(i,j)$. Then with probability $\ge 1-\exp(-r)$, we have
%\[ \sigma_{r/2}(Q(\ytil)) \ge \theta /.., \]
%where $Q(\ytil)$ is the matrix whose $(i,j)$th entry is $\iprod{\ytil, v_{(i,j)}}$.
%\end{lemma}
%The remainder of the section will focus on two questions: what do we achieve by constructing such $M_j$ (Lemma~\ref{lem:toosystem:implies})? And second, how do we construct such $M_j$? (Note that we do not require constructing them in the algorithmic sense) (Lemma~\ref{lem:toosystem:create})? 
%These two lemmas imply the proof of Proposition~\ref{prop:matrixcase}. Further, they will also play a crucial role in the inductive proof for higher order products. 
Our focus will now be on constructing such a $(\theta,\delta)$ orthogonal system of matrices, given a subspace $\calV$ of $\R^{n^2}$ of dimension $\Omega(n^2)$. The following lemma achieves this%that in such a subspace $\calV$, we can find matrices $M_1, M_2, \dots, M_r$ with a large subset of their columns having an ordered $\theta$ orthogonal property.  %We will then apply the previous lemma (Lemma~\ref{lem:toosystem:implies}) to this system. 

\begin{lemma}\label{lem:toosystem:create}
Let $\calV$ be a $\delta \cdot nm$ dimensional subspace $\R^{nm}$, and suppose $r, \theta,\delta'$ satisfy $\delta' \le \delta/2$, $r \cdot  \delta' m < \delta n/2$ and $\theta = 1/(nm^{3/2})$.  Then there exist $r$ matrices $M_1, M_2, \dots, M_r$ of dimension $n \times m$ with the following properties
\begin{enumerate}
\item $\vect(M_i) \in \calV$ for all $i \in [r]$.
\item $M_1, M_2, \dots ,M_r$ form an ordered $(\theta,\delta')$ orthogonal system.
\end{enumerate}
In particular, when $m \le \sqrt{n}$, they form an ordered $(\theta,\delta/2)$ orthogonal system. 
\end{lemma}
%\TODO{Note that we place an explicit condition of $m \le \sqrt{n}$ ! This makes the statement cleaner, and this is how we use it in all our proofs. This also motivates the restriction to a $n \times \sqrt{n}$ block.}
%\begin{remark}
%The lemma holds with a wider range of parameters $m,r$, and we prove a slightly more general version of this lemma in the appendix~\ref{app:creation}
%We will primarily use this lemma with $m =O(\sqrt{n})$ in the proof of the main theorem ~\ref{thm:kr} repeatedly. 
%\end{remark}
%\TODO{Make the following statement nicer to read like the next one, but using $(\theta,\delta)$orthogonality.}
%\begin{lemma}[Constructing ordered $(\theta,\delta)$-orthogonal system]\label{lem:toosystem:create}
%Suppose we have a $r$-dimensional subspace $\calV$ of $\R^{n \cdot m}$, for $r,n,m \in \mathbb{N}$ satisfying $n \ge m$ and $r \ge \delta n m$. Let $r' \le r$, and $\delta' \le \delta/2$ satisfy $r',\delta' m \le \sqrt{\delta n}/2$, $\delta' \le \delta/2$ and let $\theta \le \frac{1}{nm^{3/2}}$.\\
%Then, there exists $r'$ matrices $\set{Q_j}_{j \in [r']}$ of size $n \times m$ such that 
%\begin{enumerate}
%\item for all $j \in [r'], Q_j \in \text{rowspan}(P)$, $\norm{Q_j} \le 1$,
%\item $\set{Q_j}_{j \in [r']}$ are ordered $(\theta,\delta')$-orthogonal. 
%\end{enumerate}
%\end{lemma}

We remark that while $\delta$ is often a constant in our applications, $\delta'$ does not have to be.  We will use this in the proof that follows, in which we use these above two lemmas regarding construction and use of an ordered $(\theta,\delta)$-orthogonal system to prove Proposition~\ref{prop:matrixcase}.

\paragraph{Proof of Proposition~\ref{prop:matrixcase} }
The proof follows by combining Lemma~\ref{lem:toosystem:create} and Lemma~\ref{lem:toosystem:implies} in a fairly straightforward way. 
Let $\calU$ be the span of the top $\delta n^2$ singular values of $M$.  Thus $\calU$ is a $\delta n^2$ dimensional subspace of $\R^{n^2}$.  It has three steps:
\begin{enumerate}
\item We use Lemma~\ref{lem:toosystem:create} with
$m = n,  \delta' = \frac{\delta}{n^{1/2}}, \theta=\frac{1}{n^{5/2}}$ to obtain $r = \frac{n^{1/2}}{2}$ matrices  $M_1, M_2, \dots, M_r \in \R^{n \times n}$ having the $(\theta, \delta')$-orthogonality property.% with $\delta' m = \frac{n^{1/2}}{2}$.
\item Now, applying Lemma~\ref{lem:toosystem:implies}, we have that the matrix $Q(\xtil)$, defined as before, (given by linear combinations along $\xtil$) , 
has  $\sigma_{r/2}\left( Q(\xtil) \right) \ge \frac{\rho \theta}{n ^4}$ w.p $1-\exp(-\sqrt{n})$. 
\item Applying Fact~\ref{lem:sv:anticonc}  along with a simple averaging argument, we have that for one of the terms $M_i$, we have $|M_i (\xtil \otimes \ytil)| \ge \rho \theta/n^6$ with probability $\ge 1-\exp(-r/2)$ as required.
\end{enumerate}
Please refer to Appendix~\ref{app:prop10} for the complete details.
\qed

The proof for higher order tensors  will proceed  along similar lines. However we require an additional pre-processing step and a careful inductive statement (Theorem~\ref{thm:induction}), whose proof invokes Lemmas~\ref{lem:toosystem:create} and~\ref{lem:toosystem:implies}. The issues and details with higher order products are covered in Section~\ref{sec:higherproducts}.
The following two sections are devoted to proving the two lemmas i.e. Lemma~\ref{lem:toosystem:create} and Lemma~\ref{lem:toosystem:implies}. These will be key to the general case ($\ell >2$) as well.
%\subsubsection{Constructing the set of matrices -- Lemma~\ref{lem:construct-theta}}
%
%\subsubsection{Matrices with orthogonal properties suffice -- Lemma~\ref{lem:theta-suffices}}

%This section will involve repeatedly flattening out order-$\ell$ tensors in $\R^{n_1 \times n_2 \times \dots \times n_\ell}$ into vectors in $\R^{n_1 \cdot n_2 \dots n_\ell}$. We may sometimes use $T$ to represent both an order-$\ell$ tensor and a flattened vector interchangeably, based on the context.  

%The Proposition~\ref{prop:matrixcase} immediately implies Claim~\ref{claim:2d:subspaces} for the two-wise product by taking matrix $P$ with rows of $P$ being the orthonormal basis of $\calV$.
%We now give the proofs of the two main lemmas regarding ordered $(\theta,\delta)$ orthogonal systems (Lemma~\ref{lem:toosystem:create} and Lemma~\ref{lem:toosystem:implies}).
%%%%%%%%%%%%%%%%%%%%%%%%%%%%%%%%%%%%%%%%%%%%%%%%%%%%%%%%%%%%%%%%%%%%%%%%%%%%%%%%%%%%%%%%%%%%%%%%%%%%%%%%%%%
\subsection{Constructing the $(\theta,\delta)$-Orthogonal System (Proof of Lemma~\ref{lem:toosystem:create})}

Recollect that $\calV$ is a subspace of $\R^{n \cdot m}$ of dimension $\delta nm$ in Lemma~\ref{lem:toosystem:create}. We will also treat a vector $M \in \calV$ as a matrix of size $n \times m$,  with its co-ordinates indexed by $[n]\times [m]$.  

We want to construct many matrices $M_1, M_2, \dots M_r \in \R^{n \times m}$ such that a reasonable fraction of the $m$ columns satisfy $\theta$-orthogonality property. Intuitively, 
%these columns represent sets of co-ordinates (or columns) which have many degrees of freedom ---
such columns would have $\Omega(n)$ independent directions in $\R^n$, as choices for the $r$ matrices $M_1, M_2, \dots, M_r$. Hence, we need to identify columns $i \in [m]$, such that the projection of $\calV$ onto these $n$ co-ordinates (in column $i$) spans a large dimension, in a robust sense.  
%To construct an ordered $(\theta,\delta)$-orthogonal system of $r'$ matrices $\set{Q_j}$, we need to pick columns (indexed by $[m]$) which can have many linearly-independent choices for the $r'$ different matrices. Hence, we want to choose columns such that $\calV$ projected onto these columns span large dimensional subspaces of $\R^n$ (in a robust sense). 
This notion is formalized by defining the robust dimension of column projections, as follows.
\begin{definition}[Robust Dimension of projections]\label{def:dimproj}
For a subspace $\calV$ of $\R^{n\cdot m}$, we define its robust dimension $\dimt{i}{\calV}$ to be
\begin{align*}
\dimt{i}{\calV}=\max_d ~\text{s.t.} ~& \exists ~\text{orthonormal } v_1,v_2,\dots, v_d \in \R^n \text{ and } M_1,M_2, \dots, M_d \in \calV \\ 
&~\text{with } \forall t \in [d], \norm{M_t}\le \tau ~\text{and } v_t = M_t(i)  .
\end{align*}
\end{definition}
This definition ensures that we do not take into account those spurious directions in $\R^n$ that are covered to an insignificant extent by projecting (unit) vectors in $\calV$ to the $i$th column.    
\newcommand{\rhotau}{\frac{1}{\sqrt{p_2}}}
Now, we would like to use the large dimension of  $\calV$ (dim=$\delta nm$) to conclude that there are many columns projections having large robust dimensions of around $\delta n$ .
%The following lemma shows that the average robust dimension of column projections is $\delta n$ in that case. 
\begin{lemma}\label{lem:blocks}
In any subspace $\calV$ in $\R^{p_1 \cdot p_2}$ of dimension $\dim(\calV)$ for any $\tau \ge \sqrt{p_2}$, we have 
\begin{equation}
\sum_{i \in [p_2]} \dimt{i}{\calV} \ge \dim(\calV)
\end{equation} 
\end{lemma}
\begin{remark}
This lemma will also be used in the first step of the proof of Theorem~\ref{thm:kr} to identify a {\em good block} of co-ordinates which span a large projection of a given subspace $\calV$.
\end{remark}
The above lemma is easy to prove if the dimension of the column projections used is the usual dimension of a vector space. However, with robust dimension, to carefully avoid spurious or insignificant directions, we identify the robust dimension with the number of large singular values of a certain matrix.    

\begin{proof}
Let $d=\dim(\calV)$. Let $B$ be a $(p_1 p_2) \times d$ matrix, with the $d$ columns comprising an orthonormal basis for $\calV$. Clearly $\sigma_d (B)=1$. 
Now, we split the matrix $B$ into $p_1$ blocks of size $p_1 \times d$ each. For $i \in [p_2]$,  let $B_i \in \R^{p_1 \times d}$ be the projection of $B$ on the rows given by $[p_1] \times i$. 
%Hence, $B$ is obtained by just concatenating the columns of these $p_2$ matrices. 
Let $d_i=\max t$ such that $\sigma_t(B_i) \ge \rhotau$. \\
We will first show that $\sum_i d_i \ge d$. Then we will show that $\dimt{i}{\calV} \ge d_i$ to complete our proof.

Suppose for contradiction that $\sum_{i \in [p_2]} d_i < d$. Let $\calS_i$ be the $(d-d_1)$-dimensional subspace of $\R^d$ spanned by the last $(d-d_1)$ right singular vectors of $B_i$. Hence, 
$$\text{ for unit vectors} ~\alpha \in \calS_i \subseteq \R^d,~ \norm{B_i \alpha} < \rhotau .$$
Since, $d-\sum_{i \in [p_2]} d_i >0$, there exists at least one unit vector $\alpha \in \bigcap_i \calS_i^{\perp}$. Picking this unit vector $\alpha \in \R^d$, we can contradict $\sigma_d(B) \ge 1$  

%Hence $\sum_i d_i \ge d$. 
To establish the second part, consider the $d_i$ top left-singular vectors  for matrix $B_i$ ($\in \R^{p_1}$) .  These $d_i$ vectors can be expressed as small combinations ($\norm{\cdot}_2 \le \sqrt{p_2}$) of the columns of $B_i$ using Lemma~\ref{lem:smallcomb}.  
%corresponding to the top $d_i$ left-singular vectors of $B_i$. By using Lemma~\ref{lem:smallcomb}, we know that each of these $j \in [d_i]$ vectors can be expressed as a small combination $\vec{\alpha_j}$ of the columns of $B_i$ s.t. $\norm{\vec{\alpha_j}} \le \sqrt{p_2}$.
The corresponding $d_i$ small combinations of the columns of the whole matrix $B$, gives vectors in $\R^{p_1 p_2}$ which have length $\sqrt{p_2}$ as required (since column of $B$ are orthonormal). 
%Further, if we associate with each of these $j \in [d_i]$ vectors, the vector $w_j \in \R^{(p_1 p_2)}$ given by the same combination $\vec{\alpha_j}$ of the columns of $B$, we see that $\norm{w_j} \le \sqrt{p_2}$ since the columns of the matrix $B$ are orthonormal.  
\end{proof}
%%%%%%%%%%%%%%%%%%%%%%%%%%

We will construct the matrices  $M_1, M_2, \dots, M_r \in \R^{n \times m}$ in multiple stages. In each stage, we will focus on one column $i \in [m]$: we \emph{fix} this column for all the matrices $M_1,M_2,\dots, M_r$, so that this column satisfies the ordered $\theta$-orthogonal property w.r.t previously chosen columns, and then leave this column unchanged in the rest of the stages.  

In each stage $t$ of this construction we will be looking at subspaces of $\calV$ which are obtained by zero-ing out all the  columns $J \subseteq [m]$ (i.e. all the co-ordinates $[n] \times J$), that we have fixed so far. 

\begin{definition}[Subspace Projections]
For $J \subseteq [m]$, let $\Rstz{\calV}{J} \subseteq \R^{n \cdot (m-|J|)}$ represent the subspace obtained by projecting on to the co-ordinates $[n] \times ([m]-J)$, the subspace of $\calV$ having zeros on all the co-ordinates $[n]\times J$. 
\begin{equation*}
\Rstz{\calV}{J}=\set{ M' \in \R^{n \cdot (m-|J|)}: \exists M \in \calV ~\text{s.t.}~  \text{ columns } M(i) =M'(i) ~\text{ for } i \in [m]-J ,~\text{and } 0 ~\text{otherwise }.}
\end{equation*}
The extension $\extz{J}{M'}$ for $M' \in \Rstz{\calV}{J}$ is the vector $M \in \calV$ obtained by padding $M'$ with zeros in the coordinates $[n]\times J$ (columns given by $J$).   
\end{definition}
 
The following lemma shows that their dimension remains large as long as $|J|$ is not too large:
\begin{lemma}\label{lem:2d:restz}
For any $J \subseteq [m]$ and any subspace $\calV$ of $\R^{n\cdot m}$ of dimension $\delta\cdot nm$, the subspace having zeros in the co-ordinates $[n] \times J$ has  
$\dim\left(\Rstz{\calV}{J}\right)  \ge n (\delta m -|J|)$.
\end{lemma}
\begin{proof}[Proof of Lemma~\ref{lem:2d:restz}]
Consider a constraint matrix $C$ of size $(1-\delta)nm \times nm$ which describes $\calV$. $\Rstz{\calV}{J}$ is described by the constraint matrix of size $(1-\delta)nm \times n(m-|J|)$ obtained by removing the columns of $C$ corresponding to $[n]\times J$. Hence we get a subspace of dimension at least $n(m-|J|)-(1-\delta)nm$.  
\end{proof}

%We will now construct the matrices (vectors in $\R^{n\cdot m}$) $M_1, M_2, \dots, M_r$ in multiple stages. In each stage, we will focus on one column $i \in [m]$ (picked appropriately): we fix this column for each of the matrices $M_1,M_2,\dots, M_r$ so that they satisfy the required $\theta$-orthogonal property, and then leave this column unchanged in the rest of the stages. 

We now describe the construction more formally.

\rule{0pt}{12pt}
\hrule height 0.4pt
\rule{0pt}{1.5pt}
\hrule height 0.4pt
\rule{0pt}{6pt}

\noindent \textbf{The Iterative Construction of ordered $\theta$-orthogonal matrices.}

\medskip
Initially set $J_0=\emptyset$ and $M_j=0$ for all $j \in [r]$, $\tau = \sqrt{m}$ and $s=\delta m /2$. 

For $t= 1 \dots s$,
\begin{enumerate}
\item Pick $i \in [m]-J_{t-1}$ such that $\dimt{i}{\Rstz{\calV}{J_{t-1}}} \ge \delta n/2$. If no such $i$ exists, report FAIL.
\item Choose $Z_1, Z_2, \dots, Z_r \in \Rstz{\calV}{J_{t-1}}$ of length at most $\sqrt{mn}$ such that $i$th columns\\ $Z_1(i), Z_2(i), \dots, Z_r(i) \in \R^n$ are orthonormal, and also orthogonal to the columns $\{M_j(i')\}_{i' \in J_{t-1}, j \in [r]}$. If this is not possible, report FAIL. 
\item Set for all $j \in [r]$, the new $M_j \leftarrow M_j + \extz{J}{Z_j}$, where $\extz{J}{Z_j}$ is the matrix padded with zeros in the columns corresponding to $J$. Set $J_t \leftarrow J_{t-1} \cup \set{i}$.    
\end{enumerate}
\hrule height 0.4pt
\rule{0pt}{1pt}
\hrule height 0.4pt
\rule{0pt}{12pt}

Let $J=J_{s}$ for convenience. 
We first show that the above process for constructing $M_1, M_2, \dots, M_r$ completes successfully without reporting FAIL.
\begin{claim}
For $r, s$ such that $s \le \delta m/2$ and $r \cdot s \le \delta n/3$, the above process does not FAIL.
\end{claim}
\begin{proof}
In each stage, we add one column index to $J$. Hence, $|J_t|\le s$ at all times $t \in [s]$. 

We first show that Step 1 of each iteration does not FAIL. From Lemma~\ref{lem:2d:restz}, we have $\dim\left(\Rstz{\calV}{J_t}\right) \ge \delta nm/2 $. 
Let $\calW=\Rstz{\calV}{J_t}$. Now, applying Lemma~\ref{lem:blocks} to $\calW$, we see that there exists $i \in [m]-J_t$ such that $\dimt{i}{\calW} \ge \delta n/2$, as required. Hence, Step 1 does not fail.

$\dimt{i}{\calW}\ge \delta n/2$ shows that there exist $Z'_1, Z'_2, \dots Z'_{\delta n/2}$ with lengths at most $\sqrt{m}$ such that their $i$th columns $\set{Z'_t(i)}_{t\le \delta n /2}$ are orthonormal.
However, we additionally need to impose that the $i$th columns to also be orthogonal to the columns $\set{M_j(i')}_{j \in [r], i' \in J_{t-1}}$. Fortunately, the number of such orthogonality constraints is at most $r |J_{t-1}| \le \delta n/3$. Hence, we can pick the $r < \delta n/6$ orthonormal $i$th columns $\set{Z_j(i)}_{j \in [r]}$ and their respective extensions $Z_j$, by taking linear combinations of $Z'_t$. Since the linear combinations result again in unit vectors in the $i$th column, the length of $Z_j \le \sqrt{m}{n}$, as required. Hence, Step 2 does not FAIL as well.  
\end{proof}

\paragraph{Completing the proof of Lemma~\ref{lem:toosystem:create}.}

We now show that since the process completes, then $M_1,M_2, \dots, M_r$ have the required ordered $(\theta,\delta')$-orthogonal property for $\delta'=s/m$.
We first check that $M_1, M_2,\dots, M_r$ belong to $\calV$. This is true because in each stage, $\extz{J}{Z_j} \in \calV$, and hence $M_j \in \calV$ for $j \in [r]$. Further, since we run for $s$ stages, and each of the $Z_j$ are bounded in length by $\sqrt{mn}$, $\norm{M_j}_F \le s\sqrt{mn} \le \sqrt{n m^3}$.
Our final matrices $M_j$ will be scaled to $\norm{\cdot}_F=1$.  
The $s$ columns that satisfy the ordered $\theta$-orthogonality property are those of $J$, in the order they were chosen (we set this order to be $\pi$, and select an arbitrary order for the rest).
 
Suppose the column $i_t \in [m]$ was chosen at stage $t$. The key invariant of the process is that once a column $i_t$ is chosen at stage $t$, the $i_t^{th}$ column remains unchanged for each $M_j$ in all subsequent stages ($t+1$ onwards). By the construction, $Z_j(i_t) \in \R^n$ is orthogonal to $\{M_j(i)\}_{i \in J_{t-1}}$. Since $Z_j(i_t)$ has unit length and $M_j$ is of bounded length, we have the ordered $\theta$-orthogonal property as required, for $\theta=1/\sqrt{n m^3}$.    
This concludes the proof.
\qed

\subsection{$(\theta,\delta)$-Orthogonality and $\rho$-Perturbed Combinations (Proof of Lemma~\ref{lem:toosystem:implies})}\label{sec:theta-suffice}
%Let us outline the proof of the lemma. 
Suppose $M_1, M_2, \dots, M_r$ be a $(\theta, \delta)$-orthogonal set of matrices (dimensions $n \times m$).  Without loss of generality, suppose that
the permutation $\pi$ in the definition of orthogonality is the identity, and let $I$ be the first $\delta m$ columns.

Now let us consider an $\rho$-perturbed vector $\xtil$, and consider the matrix $Q(\xtil)$ defined in the statement -- it has dimensions $r \times m$, and the $(i,j)$th entry is $\iprod{ \xtil, (M_i)_j}$, which is distributed as a translated gaussian.  Now for any column $i \in I$, the $i$th column in $Q(\xtil)$ has every entry having an $(\rho \cdot \theta)$ `component' independent of entries in the previous columns, and entries above.  This implies that for a unit gaussian vector $g$, we have (by anti-concentration and $\theta$-orthogonality 
%-- see Appendix~\ref{app:too-system-implies} for details) 
that
\begin{equation}\label{eq:tmp1}
Pr[ (g^T Q(\xtil)_i)^2 < \theta^2/4n] < 1/2n.
\end{equation}
Furthermore, the above inequality holds, even {\em conditioned} on the first $(i-1)$ columns of $Q(\xtil)$. 

\begin{lemma}\label{lem:suffices-lem1}
Let $Q(\xtil)$ be defined above, and fix some $i \in I$.  Then for $g \sim \calN(0,1)^n$, we have
\[ \Pr[ (g^T Q(\xtil)_i)^2 < \frac{\theta^2 \rho^2}{4n^2}~ |~ Q(\xtil)_{1}, \dots, Q(\xtil)_{(i-1)}] < \frac{1}{2n}, \]
for any given $Q(\xtil)_1, Q(\xtil)_2, \dots, Q(\xtil)_{(i-1)}$. 
\end{lemma}
\begin{proof}
Let $g = (g_1, g_2, \dots, g_r)$.  Then we have
\begin{align*}
g^T Q_i(\xtil) &= g_1 (\xtil^T (M_1)_i) + g_2 (\xtil^T (M_2)_i) + \dots + g_r (\xtil^T (M_r)_i) \\
&= \iprod{ \xtil, g_1 (M_1)_i + g_2 (M_2)_i + \dots g_r (M_r)_i }
\end{align*}
Let us denote the latter vector by $v_i$ for now, so we are interested in $\iprod{\xtil, v_i}$.  We show that $v_i$ has a non-negligible component orthogonal to the span of $v_{1}, v_2, \dots v_{(i-1)}$.  Let $\Pi$ be the matrix which projects orthogonal to the span of $(M_s)_{i'}$ for all $i' <i$.  Thus any vector $\Pi u$ is also orthogonal to the span of $v_{i'}$ for $i'<i$.

Now by hypothesis, every vector $\Pi (M_s)_i$ has length $\ge \theta$.  Thus the vector $\Pi \left( \sum_s g_s (M_s)_i \right) = \Pi v_i$ has length $\ge \theta/2$ with probability $\ge 1-\exp(-r)$ (Lemma~\ref{claim:theta-orthog}).

Thus if we consider the distribution of $\iprod{\xtil, v_i} = \iprod{x, v_i} + \iprod{e, v_i}$, it is a one-dimensional gaussian with mean $\iprod{x, v_i}$ and variance $\rho^2$.  From basic anti-concentration properties of a gaussian (that the mass in any $\rho \cdot (\text{variance})^{1/2}$ interval is at most $\rho$), the conclusion follows.
\end{proof}

We can now do this for all $i \in I$, and conclude that the probability that Eq.~\eqref{eq:tmp1} holds {\em for all} $i \in I$ is at most $1/(2n)^{|I|}$.

Now what does this imply about the singular values of $Q(\xtil)$? Suppose it has $<r/2$ (which is $<|I|$) non-negligible singular values, then a gaussian random vector $g$, with probability {\em at least} $n^{-r}$, has a negligible component along {\em all} the corresponding singular vectors, and thus the length of $g^T Q(\xtil)$ is negligible with at least this probability!

\begin{lemma}\label{lem:sing-value-prob}
Let $M$ be a $t \times t$ matrix with spectral norm $\le 1$.  Suppose $M$ has at most $r$ singular values of magnitude $> \tau$.  Then for $g \sim \calN(0,1)^t$, we have
\[ \Pr[ \norm{Mg}_2^2 < 4t\tau^2 + \frac{t}{n^{2c}} ] \ge \frac{1}{n^{cr}} - \frac{1}{2^t}. \]
\end{lemma}
\begin{proof}
Let $u_1, u_2, \dots, u_r$ be the singular vectors corresponding to value $> \tau$.  Consider the event that $g$ has a projection of length $< 1/n^c$ onto $u_1, u_2, \dots, u_r$.  This has probability $\ge \frac{1}{n^{cr}}$, by anti-concentration properties of the Gaussian (and because $\calN(0,1)^t$ is rotationally invariant).  For any such $g$, we have
\begin{align*}
\norm{ Mg}_2^2 &= \sum_{i=1}^r \iprod{g, u_i}^2 + \tau^2 \norm{g}^2 \\
&\le \frac{r}{n^{2c}} + \tau^2 \norm{g}_2^2.
\end{align*} 
\end{proof}

%(See Lemma~\ref{lem:sing-value-prob} for details). 
This contradicts the earlier anti-concentration bound, and so we conclude that the matrix has at least $r/2$ non-negligible singular values, as required.  
%We refer to Appendix~\ref{app:too-system-implies} for the details.

%Recall that the $(i,j)$th entry of $Q(\xtil)$ is $\iprod{ \xtil, (M_i)_j}$, which is distributed as a translated gaussian, since $\xtil = x + e$, for $e \sim \calN(0, \rho^2)^n$.

%%%%%%%%%%%%%%%%%%%%%%%%%%%%%%%%%%%%%%%%%%%%%%%%%%%%%%%%%%%%%%%%%%%%%%%%%%%%%%%%%%%%%%%%%%%%%%%%%%%

\subsection{Higher Order Products} \label{sec:higherproducts}

%\subsection{Higher Order Products} \label{app:higherproducts}

We have a subspace $\calV \in \R^{n^\ell}$ of dimension $\delta n^{\ell}$. 
The proof for higher order products proceeds by induction on the order $\ell$ of the product. Recall from Remark~\ref{rem:baddelta} that Proposition~\ref{prop:matrixcase} and Theorem~\ref{thm:krl} do not get good  guarantees for small values of $\delta$, like $1/n$. In fact, we can not hope to get such exponentially small failure probability in that case, since the all the $n$ degrees of freedom in $\calV$ may be constrained to the first $n$ co-ordinates  of $\R^{n^2}$ (all the independence is in just one mode). Here, it is easy to see that the best we can hope for is an inverse-polynomial failure probability. Hence, to get exponentially small failure probability, we will always need $\calV$ to have a large dimension compared to the dimension of the host space in our inductive statements.

To carry out the induction, we will try to reduce this to a statement about $\ell-1$ order products, by taking linear combinations (given by $\xtil^{(1)} \in \R^n$) along one of the modes.  Loosely speaking, Lemma~\ref{lem:toosystem:implies} serves this function of ``order reduction'', however it needs a set of $r$ matrices in $\R^{n \times m}$ (flattened along all the other modes) which are ordered $(\theta,\delta)$ orthogonal. 

Let us consider the case when $\ell=3$, to illustrate some of the issues that arise. We can use Lemma~\ref{lem:toosystem:create} to come up with $r$ matrices in $\R^{n \times n^2}$ that are ordered $(\theta,\delta)$ orthogonal. These columns intuitively correspond to independent directions or degrees of freedom, that we can hope to get substantial projections on. However, since these are vectors in $\R^n$, the number of ``flattened columns'' can not be comparable to $n^2$ (in fact, $\delta m \ll n$) --- hence, our induction hypothesis for $\ell=2$ will give no guarantees, (due to Remark~\ref{rem:baddelta}).

To handle this issue, we will first restrict our attention to a smaller block of co-ordinates of size $n_1 \times n_2 \times n_3$ (with $n_1 n_2 n_3 \ll n$) , that has reasonable size in all the three modes ($n_1, n_2, n_3 = n^{\Omega(1)}$). % $n^{\Omega(1)} \le n_1, n_2, n_3 \ll n^{1/3}$, 
Additionally, we want $\calV$'s projection onto this $n_1 \times n_2 \times n_3$ block spans a large subspace of (robust) dimension at least $\delta n_1 n_2 n_3$ (using Lemma~\ref{lem:blocks}). 

Moreover, choosing the main inductive statement also needs to be done carefully. We need some property for choosing enough candidate ``independent'' directions $T_1, T_2, \dots T_r \in \R^{n^\ell}$ (projected on the chosen block), such that our process of ``order reduction'' (by first finding $\theta$-orthogonal system and then combining along $\xtil^{(1)}$) maintains this property for order $\ell-1$. This is where the alternate interpretation in Theorem~\ref{thm:restatement} in terms of singular values helps: it  suggests the exact property that we need! We ensure that the matrix formed by the flattened vectors $\vect(T_1),\vect(T_2), \dots \vect(T_r)$ (projected onto the $n_1 \times n_2 \times n_3$ block) , as rows form a matrix with many large singular values.

We now state the main inductive claim. The claim assumes a block of co-ordinates of reasonable size in each mode that span many directions in $\calV$, and then establishes the anti-concentration bound inductively.
\newcommand{\nn}{N}
\newcommand{\del}{\delta}

\begin{theorem}[Main Inductive Claim] \label{thm:induction}
Let $T_1, T_2, \dots, T_r \in \R^{n^{\times \ell}}$ be $r$ tensors with bounded norm ($\norm{\cdot}_F \le 1$) and $I_1, I_2, \dots I_\ell \subseteq [n]$ be sets of indices of sizes $n_1, n_2, \dots n_\ell$. Let $T$ be the $r \times n^{\ell}$ matrix obtained with rows $\vect(T_1), \vect(T_2), \dots, \vect(T_r)$.  Suppose
\begin{itemize}
\item $\forall j \in [r], P_j$ is $T_j$ restricted to the block $I_1  \times \dots \times I_{\ell}$, and matrix $P \in \R^{r \times (n_1 \cdot n_2 \dots n_\ell)}$ has $j$th row as $\vect(P_j)$,
\item $r \ge \delta_{\ell} n_1 n_2 \dots n_{\ell}$  and $\forall t \in [\ell-1], n_{t} \ge \left( n_{t+1} n_{t+2} \dots n_\ell \right)^2$, 
\item $\sigma_{r}(P) \ge \eta$.
\end{itemize}
Then for random $\rho$-perturbations $\xtil^{(1)}, \xtil^{(2)}\dots \xtil^{(\ell)}$ of any $x^{(1)}, x^{(2)}\dots x^{(\ell)} \in \R^n$, we have
$$ \Pr_{\xtil^{(1)}, \dots \xtil^{(\ell)}} \left[ \norm{T \left(\xtil^{(1)} \otimes \dots \otimes \xtil^{(\ell)}\right)}\ge \rho^\ell \left(\frac{\eta}{n_1}\right)^{3^\ell} \right] \ge 1- \exp\left(- \delta_\ell n_\ell\right)$$
\end{theorem}

Before we give a proof of the main inductive claim, we first present  a standard fact that relates the singular value of matrices and some anti-concentration properties of randomly perturbed vectors. This will also establish the base case of our main inductive claim.

\begin{fact}\label{lem:sv:anticonc}
Let $M$ be a matrix of size $m \times n$ with $\sigma_r(M) \ge \eta$. Then for any unit vector $u \in \R^n$ and an random $\rho$-perturbation $\util$ of it, we have
$$ \norm{M \util}_2 \ge \eta \rho/ n^2 ~\text{ w.p }~ 1-n^{-\Omega(r)}$$
\end{fact}

\begin{proof}[Proof of Theorem~\ref{thm:induction}]
The proof proceeds by induction. The base case ($\ell=1$) is handled by Fact~\ref{lem:sv:anticonc}. Let us assume the theorem for $(\ell-1)$-wise products. The inductive proof will have two main steps:
\begin{enumerate}
\item Suppose we flatten the tensors $\set{P_j}_{j \in [r]}$ along all but the first mode, and imagine them as matrices of size $n_1 \times (n_2 n_3 \dots n_\ell)$. 
We can use Lemma~\ref{lem:toosystem:create} to construct ordered $(\theta,\delta')$ orthogonal system w.r.t vectors in $\R^{n_1}$ (columns correspond to $[m]=[n_2 \dots n_\ell]$).
\item When we take combinations along $\xtil^{(1)}$ as $T\left(\xtil^{(1)}, \cdot, \cdot,\dots,\cdot\right)$, these tensors will now satisfy the condition required for $(\ell-1)$-order products in the inductive hypothesis, because of Lemma~\ref{lem:toosystem:implies}.
\end{enumerate}
Unrolling this induction allows us to take combinations along $\xtil^{(1)}, \xtil^{(2)}, \dots$ as required, until we are left with the base case.  For notational convenience, let $y=\xtil^{(1)}$, $\del_\ell=\del$, $r_\ell=r$ and $\nn=n_1 n_2 \dots n_\ell$. 

To carry out the first step, we think of $\set{P_j}_{j \in [r]}$ as matrices of size $n_1 \times (n_2 n_3 \dots n_\ell)$. We then apply Lemma~\ref{lem:toosystem:create} with $n=n_1$, $m=\frac{\nn}{n_1}=n_2 n_3 \dots n_\ell \le \sqrt{n_1}$ ; hence there exists $r'_\ell=n_2 \dots n_\ell$ matrices $\set{Q_{q}}_{q \in [r'_\ell]}$ with $\norm{\cdot}_F\le 1$ which are ordered $(\theta,\delta'_\ell)$-orthogonal for $\delta'_\ell=\delta_\ell/3$. Further, since $Q_{q}$ are in the row-span of $P$, there exists matrix of coefficients $\Alpha=\left(\alpha(q,j)\right)_{q \in [r'_\ell],j\in [r_\ell]}$ such that
\begin{align} \label{eq:smallcomb}
\forall q\in [r'_\ell], ~& Q_{q}= \sum_{j \in [r_\ell]} \alpha(q,j)P_j \\
& \norm{\alpha(q)}_2^2 = \sum_{j \in [r_\ell]} \alpha(q,j)^2 \le 1/ \eta ~\quad (\text{since } \sigma_r(P) \ge \eta ~\text{and} ~ \norm{Q_q}_F \le 1)    
\end{align}
Further, $Q_{q}$ is the projection of $\sum_{j \in [r_\ell]} \alpha_{q,j} T_j$ onto co-ordinates $I_1 \times I_2 \dots \times I_\ell$. Suppose we define a new set of matrices$\set{W_{q}}_{q \in [r'_\ell]}$ in $\R^{n \times (\frac{\nn}{n_1})}$ by flattening the following into a matrix with $n$ rows: 
$$ W_q= \left( \sum_{j \in [r]} \alpha_{q,j} T_j \right)_{[n]\times \left(I_2 \times \dots \times I_\ell \right)}.$$
%Hence, the $n \times (n_2 \cdot n_3\dots n_\ell)$ matrices $\set{W_{j'}}_{j' \in [r']}$ obtained by flattening the restrictions $\left( \sum_{j \in [r]} \alpha_{j',j} T_j \right)_{[n]\times I_2 \times \dots \times I_\ell}$ 
In other words, $Q_{q}$ is obtained by projecting $W_{q}$ on to the $n_1$ rows given by $I_1$.
Note that $\set{W_q}_{q \in [r'_\ell]}$ is also ordered $(\theta'_\ell,\delta'_\ell)$ orthogonal for $\theta'_\ell=\theta \eta$. 

 To carry out the second part, we apply Lemma~\ref{lem:toosystem:implies} with $\set{W_{q}}$ and infer that the $r'_\ell \times (\nn/n_1)$ matrix $W(y)$ with $q$th row being $\transpose{y}W_{q}$ has $\sigma_{r_{\ell-1}} \left(W(y) \right) \ge \eta'_{\ell}= \theta^2 \rho^2 /n_1^4$ with probability $1-\exp(-\Omega(r'_\ell))$, where $r_{\ell-1}=r'_\ell/2$.\\
We will like to apply the inductive hypothesis for $(\ell-1)$ with $P$ being $W(y)$; however $W(y)$ {\em does not} have full (robust) row rank. Hence we will consider the top $r_{\ell-1}$ right singular vectors of $W(y)$ to construct an $r_{\ell-1}$ tensors of order $\ell$, whose projections to the block $I_2 \times  \dots \times  I_\ell$, lead to a well-conditioned $r_{\ell-1} \times (n_2 n_3 \dots n_{\ell})$ matrix  for which our inductive hypothesis holds.

Let the top $r_{\ell-1}$ right singular vectors of $W(y)$ be $Z_1, Z_2, \dots Z_{r_{\ell-1}}$. Hence, from Lemma~\ref{lem:smallcomb}, we have a coefficient $\Beta$ of size $r_{\ell-1} \times r_\ell$ such that  
$$ \forall j' \in [r_{\ell-1}] ~ \quad  Z_{j'} = \sum_{q \in [r'_\ell]} \beta_{j',q} W_{q}\left(y\right) ~\text{and}~ \norm{\beta(j')}_2 \le 1/\eta'_{\ell}.$$
%Let $Z_1,Z_2 ,\dots Z_{r_{\ell-1}} \in \R^{n \times (\nn/n_1)}$ be defined by 
%$$\forall j' \in [r_{\ell-1}],~ Z_{j'}=\sum_{q in [r'_\ell]} \beta_{j', q} W_{q} .$$
%Hence, $\set{Z_{j'}(y): j' \in [r_{\ell-1}]}$ are the top $r_{\ell-1}$ singular values of $W(y)$.   
%Since $\sigma_{r_{\ell-1}}\left(W(y)\right)\ge \eta_{\ell-1}$, the properties of SVD (lemma~
% that there exists such $\Beta$ with row norm at most $1/\eta'_{\ell}$ i.e. $\forall j' \in [r_{\ell-1}], $. 
Now let us try to represent these new vectors in terms of the original row-vectors of $P$, to construct the required tensor of order $(\ell-1)$ . 
Consider the $r_{\ell-1} \times r_\ell$ matrix $\Lambda=\Beta \Alpha$. Clearly,
$$ \text{rownorm}(\Lambda) \le \text{rownorm}(\Beta) \cdot \norm{\Alpha}_F \le \sqrt{r'
_{\ell}} \cdot \text{rownorm}(\Beta)\cdot  \text{rownorm}(\Alpha) \le \frac{r'_{\ell}}{\eta_{\ell} \eta'_{\ell}}  .$$
Define $\forall j' \in [r_{\ell-1}]$, an order $\ell$ tensor $T'_{j'} =\sum_{j \in [r]} \lambda_{j',j} T_j$; from the previous equation, $\norm{T'_{j'}}_F \le r'_\ell/(\eta_\ell \eta'_\ell)$ . We need to get a normalized order $(\ell-1)$ tensor: so, we consider $\widehat{T}_{j'}=T'_{j'}/\norm{T'_{j'}(y)}_F$, and $\widehat{T}$ be the $r_{\ell-1} \times (n^{\ell})$ matrix with $j'$th row being $\widehat{T}_{j'}$. Hence,
$$\sigma_{r_{\ell-1}} \left(\widehat{T}(y,\cdot,\cdot,\dots,\cdot)\right) \ge \frac{\eta^3_{\ell}}{r'_{\ell} n_1^3}. $$ 
We also have $r_{\ell-1} \ge \frac{1}{2}  \cdot n_2 n_3 \dots n_\ell$. By the inductive hypothesis 
\begin{equation}
\norm{\widehat{T}\left(y,\xtil^{(2)},\dots,\xtil^{(\ell)}\right)} \ge \eta' \equiv \rho^{\ell-1} \left(\frac{\eta^3_\ell}{n_1^4 n_2} \right)^{3^{\ell-1}} ~\text{ w.p}~ 1-\exp\left(-\Omega(n_\ell)\right)
\end{equation} 

Hence, for one of the $j' \in [r_{\ell-1}]$, $\abs{\widehat{T}_{j'}\left(\xtil^{(1)},\xtil^{(2)},\dots \xtil^{(\ell)}\right)} \ge \eta'/\sqrt{r_{\ell-1}}$. Finally, since $\widehat{T}_{j'}$ is given by a small combination of the $\set{T_j}_{j \in [r]}$, we have from Cauchy-Schwartz
$$\norm{T\left(\xtil^{(1)},\xtil^{(1)},\dots,\xtil^{(1)}\right)} \ge \eta' \cdot \left(\frac{\eta^3 }{\sqrt{r^2_\ell n_1^4}}\right).$$
%${j'' \in [s]}$ of size $n_2 \times n_3 \dots \times n_\ell$ given by the top $s$ right singular vectors of $W(\xtil^{(1)})$.  
%$\forall j' \in [r']$, the $n_2 \times n_3 \times \dots \times n_\ell$ tensor $T'_{j'}=\sum_{j \in [r]} \alpha_{j',j} T_j$.\\
%Note that $Q'_{j'}=(T_{j'})_{\vert I_1 \times I_2 \dots \times I_\ell}$. 
\end{proof}

The main required theorem now follows by just showing the exists of the $n_1 \times n_2 \times \dots \times n_\ell$ block that satisfies the theorem conditions. This follows from Lemma~\ref{lem:blocks}.

\begin{proof}[Proof of Theorem~\ref{thm:kr}]
First we set $n_1, n_2, n_\ell$ by the recurrence $\forall t \in [\ell], ~ n_t =2(n_{t+1}\cdot n_{t+2} \dots n_\ell)^2$ and $n_1=O(n)$. It is easy to see that this is possible for $n_\ell = n^{1/3^\ell}$. 
Now, we partition the set of co-ordinates $[n]^{\ell}$ into blocks of size $n_1 \times n_2 \times \dots n_\ell$. Let $p_1=n_1 \cdot n_2 \dots n_\ell$ and $p_2 = n^{\ell}/p_1$. Applying Lemma~\ref{lem:blocks} we see that there exists indices $I_1, I_2, \dots I_\ell$ of sizes $n_1,n_2,\dots,n_\ell$ respectively such that projection $\calW=\calV_{\vert I_1 \times I_2 \times \dots \times I_\ell}$ on this block of co-ordinates has dimension $\dimt{I}{\calW} \ge  n_1 n_2 \dots n_\ell/4$. Let $r=n_1 n_2 \dots n_\ell$. Now we construct $P'$ with the rows of $P'$ being an orthonormal basis for $\calW$, and let $T'$ be the corresponding vectors in $\calV$. Note that $\forall j \in [r], ~\norm{T'_j} \le n^{\ell}$. Let $P$ be the re-scaling of the matrix so that for the $j$th row($j \in [r]$), $P_j = P'_j / \norm{T'_j}$ and $T_j = T'_j/ \norm{T'_j}$. Hence $\sigma_r (P) \ge 1/n^{\ell}$. Applying Theorem~\ref{thm:induction} with this choice of $P,T$, we get the required result.   
\end{proof}

%\begin{lemma}[Small Combinations]\label{lem:random:sv}
%Consider an $r \times n$ matrix $M$ such that $\sigma_r(M) \ge \eta$, and a vector $y \in \R^n$ in the row-space of $M$ i.e. $y = \xtil^t M$. Any such vector $y$ of bounded norm can be expressed using a small combination of the rows: 
%\begin{equation}
%\norm{y} \le \gamma \implies \norm{x} \le \gamma/ \eta.
%\end{equation} 
%\end{lemma}

%\begin{lemma}[Small Combinations]\label{lem:random:sv}
%Consider an $r \times n$ matrix $M$ such that $\sigma_r(M) \ge \eta$, and a vector $y \in \R^n$ in the row-space of $M$ i.e. $y = \xtil^t M$. Any such vector $y$ of bounded norm can be expressed using a small combination of the rows: 
%\begin{equation}
%\norm{y} \le \gamma \implies \norm{x} \le \gamma/ \eta.
%\end{equation} 
%\end{lemma}

%
\section{Learning Multi-view Mixture Models}\label{sec:multi view}

We now see how Theorem~\ref{thm:main} immediately gives efficient learning algorithms for broad class of discrete mixture models called multi-view models in the {\em over-complete} setting. In a multi-view mixture model, for each sample we are given a few different observations or views $\spc{x}{1}, \spc{x}{2}, \dots,\spc{x}{\ell}$ that are conditionally independent given which component $i \in [R]$ the sample is from. Typically, the $R$ components in the mixture are discrete distributions. 
Multi-view models are very expressive, and capture many well-studied models like Topic Models \cite{AHK}, Hidden Markov Models (HMMs) \cite{MR,AMR,AHK},  and random graph mixtures \cite{AMR}. They are also sometimes referred to as  {\em finite mixtures of finite measure products}\cite{AMR} or {\em mixture-learning with multiple snapshots} \cite{RSS}.  

In this section, we will assume that each of the components in the mixture is a discrete distribution with support of size $n$. %Please see \cite{AMR08,AHK12} for other useful settings of multi-view models. 
We first introduce some notation, along the lines of \cite{AHK}.

\paragraph{Parameters and the model:} Let the $\ell$-view mixture model be parameterized by a set of $\ell$ vectors in $\R^n$ for each mixture component,  $\set{\spc{\mu_i}{1}, \spc{\mu_i}{2},\dots, \spc{\mu_i}{\ell}}_{i \in [R]}$, and mixing weights $\set{w_i}_{i \in [R]}$ , that add up to $1$.  Each of these parameter vectors are normalized : in this work, we will assume that $\norm{\spc{\mu_i}{j}}_1 =1$ for all $i \in [R], j \in [\ell]$.  
%This is to capture the canonical example for such models, where each $\spc{\mu_i}{j}$ represents a probability distribution over discrete distribution of support $n$. However, this normalization is not crucial.   
Finally, for notational convenience we think of the parameters are represented by $n \times R$ matrices (one per view) $\spc{M}{1}, \spc{M}{2},\dots,\spc{M}{\ell}$, with $\spc{M}{j}$ formed by concatenating the vectors $\spc{\mu_i}{j}$ ($1 \le i \le R$). 
 %In many settings, the $n$-dimensional vectors $\spc{x}{j}$ are actually indicator vectors (hence $\Bd=1$): this is commonly used to encode the case when the observation is one of $n$ discrete events.  
\\
 
\noindent Samples from the multi-view model with $\ell$ views are generated as follows:  
\begin{enumerate}
\item The mixture component $i$ ($i \in [R]$) is first picked with probability $w_i$  
\item The views $\spc{x}{1},\dots, \spc{x}{j},\dots, \spc{x}{\ell}$ are indicator vectors in $n$-dimensions, that are drawn according to the distribution $\spc{\mu_i}{1},\dots,\spc{\mu_i}{j},\dots,\spc{\mu_i}{\ell}$.
%are random vectors $\in \R^n$ that are conditionally independent given the component $i$, with means $\spc{\mu_i}{j} \in \R^n$ i.e.
%$$\E{\spc{x}{j} \vert \text{mixture} =i}= \spc{\mu}{j}_i \text{ and } \E{\spc{x}{i} \otimes \spc{x}{j}\vert \text{mixture}=r}= \spc{\mu}{j}_1 \otimes \spc{\mu}{j}_\ell $$
\end{enumerate}

%\TODO{Not sure how much of the following should be moved to the introduction?}\\
The state-of-the-art algorithms for learning multi-view mixture models have guarantees that mirror those for mixtures of gaussians. In the worst case, the best known algorithms for this problem are from a recent work Rabani et al \cite{RSS}, who give an algorithm that has complexity $R^{O(R^2)}+ \poly(n,R)$. In fact they also show a sample complexity lower-bound of $\exp(\tilde{\Omega}(R))$ for learning multi-view models in one dimension ($n=1$). Polynomial time algorithms were given by Anandkumar et al. \cite{AHK} in a restricted setting called the non-singular or non-degenerate setting. When each of these matrices $\set{\spc{M}{j}}_{j \in [\ell]}$ to have rank $R$ in a robust sense i.e. $\sigma_R(\spc{M}{j})\ge 1/\tau$ for all $j \in [\ell]$, their algorithm runs in just $\poly(R,n,\tau,1/\eps)$) time to learn the parameters up to error $\eps$. However, their algorithm fails even when $R=n+1$. 

However, in many practical settings like speech recognition and image classification, the dimension of the feature space is typically much smaller than the  number of components or clusters i.e. $n \ll R$. To the best of our knowledge, there was no efficient algorithm for learning multi-view mixture models in such \emph{over-complete settings}. We now show how Theorem~\ref{thm:main} gives a polynomial time algorithm to learn multi-view mixture models in a smoothed sense, even in the over-complete setting $R \gg n$.

\begin{theorem}\label{thm:main:multiview}
Let $(w_i, \spc{\mu_i}{1}, \dots,\spc{\mu_i}{\ell})$ be a mixture of $R = O(n^{\ell/2-1})$ multi-view models with $\ell$ views, and suppose the means $\left(\spc{\mu_i}{j}\right)_{i \in [R], j \in [\ell]}$ are perturbed independently by gaussian noise of magnitude $\rho$. Then there is a polynomial time algorithm to learn the weights $w_i$, the perturbed parameter vectors  $\set{ \spc{\tilde{\mu}_i}{j} }_{j \in [\ell], i \in [R]}$ up to an accuracy $\eps$ when given samples from this distribution. The running time and sample complexity is $\poly_\ell(n, 1/\rho, 1/\eps)$.
\end{theorem}

The conditional independence property is very useful in obtaining a higher order tensor, in terms of the hidden parameter vectors that we need to recover. This allows us to use our results on tensor decompositions from previous sections.
\begin{lemma}[\cite{AMR}]
In the notation established above for multi-view models, $\forall \ell \in \N$ the $\ell^{th}$ moment tensor
\begin{equation}\label{eq:mv:moml}
\text{Mom}_\ell=\E\left[\spc{x}{1} \otimes \dots \spc{x}{j} \otimes \dots \spc{x}{\ell} \right] = \sum_{r \in [R]} w_r \spc{\mu}{1}_r \otimes \spc{\mu}{2}_r \dots \otimes \spc{\mu}{j}_r \otimes \dots \otimes \spc{\mu}{\ell}_r.
\end{equation}
%In our usual representation of tensor decompositions, 
%$$ \E{\spc{x}{1} \otimes \dots \spc{x}{j} \otimes \dots \spc{x}{\ell}} = \left[\spc{M}{1} ~ \spc{M}{2} ~ \dots ~ \spc{M}{\ell} \right].$$
\end{lemma}

Our algorithm to learn multi-view models consists of three steps:
\begin{enumerate}
\item Obtain a good empirical estimate $\widehat{T}$ of the order $\ell$ tensor $\text{Mom}_\ell$ from $N=\poly_\ell(n,R,1/\rho,1/\eps)$ samples (given by Lemma~\ref{app:sampling:mv})
 $$ \widehat{T}=\frac{1}{N} \sum_{t =1}\spc{x_t}{1} \otimes \spc{x_t}{2} \otimes \dots \otimes \spc{x_t}{\ell}.$$
\item Apply Theorem~\ref{thm:main} to $\widehat{T}$ and recover the parameters $\spc{\widehat{\mu}_i}{j}$ upto scaling.
\item Normalize the parameter vectors $\spc{\widehat{\mu_i}}{j}$ to having $\ell_1$ norm of $1$, and hence figure out the weights $\widehat{w}_i$ for $i \in [R]$. 
\end{enumerate}

\begin{proof}[Proof of Theorem~\ref{thm:main:multiview}]
The proof follows from a direct application of Theorem~\ref{thm:main}. Hence, we just sketch the details. We first obtain a good empirical estimate of $\text{Mom}_\ell$ that is given in equation~\eqref{eq:mv:moml} using Lemma~\ref{app:sampling:mv}.
Applying Theorem~\ref{thm:main} to $\widehat{T}$, we recover each rank-$1$ term in the decomposition $w_i \spc{\mu_i}{1} \otimes \spc{\mu_i}{2}\otimes \dots \otimes \spc{\mu_i}{\ell}$ up to error $\eps$ in frobenius norm ($\norm{\cdot}_F$). However, we know that each of the parameter vectors are of unit $\ell_1$ norm. Hence, by scaling all the parameter vectors to unit $\ell_1$ norm, we obtain all the parameters up to the required accuracy.
\end{proof}
\section{Learning Mixtures of Axis-Aligned Gaussians}\label{sec:gaussians}
\newcommand{\tilmu}{\widetilde{\mu}}
Let $F$ be a mixture of $k = \mbox{poly}(n)$ axis-aligned Gaussians in $n$ dimensions, and suppose further that the means of the components are perturbed by Gaussian noise of magnitude $\rho$. We restrict to Gaussian noise not because our results change, but for notational convenience. 

\paragraph{Parameters:} The mixture is described by a set of $k$ mixing weights $w_i$, means $\mu_i$ and covariance matrices $\Sigma_i$. Since the mixture is axis-aligned, each covariance $\Sigma_i$ is diagonal and we will denote the $j^{th}$ diagonal of $\Sigma_i$ as $\sigma_{ij}^2$.  Our main result in this section is the following:

\begin{theorem}\label{thm:main:gaussians}
Let $(w_i, \mu_i, \Sigma_i)$ be a mixture of $k = n^{\lfloor \frac{\ell - 1}{2} \rfloor}/(2 \ell)$ axis-aligned Gaussians and suppose $\set{\tilmu_i}_{i \in [k]}$ are the $\rho$-perturbations of $\set{\mu_i}_{i \in [k]}$ (that have polynomially bounded length). Then there is a polynomial time algorithm to learn the parameters $(w_i, \tilmu_i, \Sigma_i)_{i \in [k]}$ up to an accuracy $\eps$ when given samples from this mixture. The running time and sample complexity is $\text{poly}_\ell(\frac{n}{\rho \eps})$.
\end{theorem}
%\vnote{We need to look at $\ell$ and $\ell+1$ order tensors in this case.}

Next we outline the main steps in our learning algorithm:
\begin{enumerate}
\item We first pick an appropriate $\ell$, and estimate $\calM_\ell := \sum_i w_i \tilmu_i^{\otimes \ell}$.\footnote{We do not estimate the entire tensor, but only a relevant ``block'', as we will see.}
\item We run our decomposition algorithm for {\em overcomplete} tensors on $\calM_\ell$ to recover $\tilmu_i, w_i$.
\item We then set up a system of linear equations and solve for $\sigma_{ij}^2$.
\end{enumerate}

\noindent We defer a precise description of the second and third steps to the next subsections (in particular, we need to describe how we obtain $\calM_\ell$ from the moments of $F$ and we need to describe the linear system that we will use to solve for $\sigma_{ij}^2$). 

\subsection{Step $2$: Recovering the Means and Mixing Weights}
Our first goal in this subsection is to construct the tensor $\calM_\ell$ defined above from random samples. In fact, if we are given many samples we can estimate a related tensor (and our error will be an inverse polynomial in the number of samples we take). Unlike the multi-view mixture model, we do not have $\ell$ independent views in this case. Let us consider the tensor $\E [x^{\otimes \ell}]$:
\[ \E [x^{\otimes \ell}] = \sum_i w_i (\tilmu_i+\eta_i)^{\otimes \ell}. \]
Here we have used $\eta_i$ to denote a Gaussian random variable whose mean is zero and whose covariance is $\Sigma_i$. Now the first term in the expansion is the one we are interested in, so it would be nice if we could ``zero out'' the other terms. Our observation here is that if we restrict to $\ell$ distinct indices $(j_1, j_2, \dots, j_\ell)$, then this coordinate will only have contribution from the means.  To see this, note that the term of interest is
\[ \sum_i \big[ w_i \prod_{t=1}^\ell (\tilmu_i (j_t) + \eta_i (j_t)) \big] \] 
Since the Gaussians are axis aligned, the $\eta_i(j_t)$ terms are independent for different $t$, and each is a random variable of zero expectation.  Thus the term in the summation is precisely $\sum_i w_i \prod_{t=1}^\ell \tilmu_i (j_t)$.

Our idea to estimate the means is now the following: we partition the indices $[n]$ into $\ell$ roughly equal parts $S_1, S_2, \dots, S_\ell$, and estimate a tensor of dimension $|S_1|\times |S_2| \times \dots \times |S_\ell|$. 
\begin{definition}[Co-ordinate partitions]
Let $S_1, S_2, \dots, S_\ell$ be a partition of $[n]$ into $\ell$ pieces of equal size (roughly). Let $\tilmu_i^{(t)}$ denotes the vector $\tilmu_i$ restricted to the coordinates $S_t$, and for a sample $x$, let $x^{(t)}$ denote its restriction to the coordinates $S_t$.   
\end{definition}

Now, we can estimate the order $\ell$ tensor $\E [x^{(1)}\otimes x^{(2)} \dots \otimes x^{(\ell)}]$ to any inverse polynomial accuracy using polynomial samples (see Lemma~\ref{lem:sampling:gaussians} or \cite{HK} for details), where
\[\E [x^{(1)}\otimes x^{(2)} \dots \otimes x^{(\ell)} ]= \sum_i w_i \big( \tilmu_i^{(1)} \otimes \tilmu_i^{(2)} \otimes \dots \otimes \tilmu_i^{(\ell)} \big). \]
%Here $\tilmu_i^{(t)}$ denotes the vector $\tilmu_i$ restricted to the coordinates $S_t$. 

 Now applying the main tensor decomposition theorem (Theorem~\ref{thm:main}) to this order $\ell$ tensor, we obtain a set of vectors $\nu_i^{(1)}, \nu_i^{(2)}, \dots, \nu_i^{(t)}$ such that
\[ \nu_i^{(t)} = c_{it} \tilmu_i^{(t)}, \text{ and for all $t$, } c_{i1} c_{i2} \cdots c_{i\ell} = 1/w_i. \]

Now we show how to recover the means $\tilmu_i$ and weights $w_i$.

\begin{claim}
The algorithm recovers the perturbed means $\set{\tilmu_i}_{i \in [R]}$ and weights $w_i$ up to any accuracy $\eps$ in time $\text{poly}_\ell (n, 1/\eps)$
\end{claim}

So far, we have portions of the mean vectors, each scaled differently (upto some $\eps/\text{poly}_\ell(n)$ accuracy. We need to estimate the scalars $c_{i1}, c_{i2}, \dots, c_{i\ell}$ up to a scaling (we need another trick to then find $w_i$). To do this, the idea is to take a different partition of the indices $S_1', S_2', \dots, S_\ell'$, and `match' the coordinates to find the $\tilmu_i$.  In general, this is tricky since some portions of the vector may be zero, but this is another place where the perturbation in $\tilmu_i$ turns out to be very useful (alternately, we can also apply a random basis change, and a more careful analysis to doing this 'match').

\begin{claim}
Let $\mu$ be any $d$ dimensional vector. Then a coordinate-wise $\sigma$-perturbation of $\mu$ has length $\ge d\sigma^2/10$ w.p. $\ge 1-\exp(-d)$.
\end{claim}

The proof is by a basic anti-concentration along with the observation that coordinates are independently perturbed and hence the failure probability multiplies.

\iffalse
\begin{claim}
Let $\mu$ and $\nu$ be two vectors with $d$ coordinates each, and suppose $\mu$ is perturbed independently in each coordinate by $\calN(0, \sigma^2)$ to obtain $\widehat{\mu}$. Then
\[ \min_c \norm{ \widehat{\mu} - c\nu }^2 \ge d\sigma^2/10 \text{ with probability } \ge 1- \exp(-d). \]
\end{claim}
The lemma says that a perturbed vector is unlikely to be ``parallel'' to any given vector. This implies that two independently perturbed vectors cannot be parallel to each other w.p. $\ge 1-\exp(-d)$.

\begin{proof}
A $\sigma$-perturbed vector has a $d\sigma^2/10$ projection on to any large enough dimensional space, in particular the space orthogonal to $\nu$.
\end{proof}
\fi

Let us now define the partition $S_t'$.  Suppose we divide $S_1$ and $S_2$ into two roughly equal parts each, and call the parts $A_1, B_1$ and $A_2, B_2$ (respectively). Now consider a partition with $S_1' = A_1 \cup A_2$ and $S_2' = B_1 \cup B_2$, and $S_t' = S_t$ for $t >2$.  Consider the solution $\nu_i'$ we obtain using the decomposition algorithm, and look at the vectors $\nu_1, \nu_2, \nu_1', \nu_2'$. For the sake of exposition, suppose we did not have any error in computing the decomposition.  We can scale $\nu_1'$ such that the sub-vector corresponding to $A_1$ is precisely equal to that in $\nu_1$.  Now look at the remaining sub-vector of $\nu_1$, and suppose it is $\gamma$ times the ``$A_2$ portion'' of $\nu_2$.  Then we must have $\gamma = c_2/c_1$.

To see this formally, let us fix some $i$ and write $v_{11}$ and $v_{12}$ to denote the sub-vectors of $\tilmu_i^{(1)}$ restricted to coordinates in $A_1$ and $B_1$ respectively.  Write $v_{21}$ and $v_{22}$ to represent sub-vectors of $\tilmu_i^{(2)}$ restricted to $A_2$ and $B_2$ respectively. Then $\nu_1$ is $c_1 v_{11} \oplus c_1 v_{12}$ (where $\oplus$ denotes concatenation). So also $\nu_2$ is $c_2 v_{21} \oplus c_2 v_{22}$. Now we scaled $\nu_1'$ such that the $A_1$ portion agrees with $\nu_1$, thus we made $\nu_1'$ equal to $c_1 v_{11} \oplus c_1 v_{21}$.  Thus by the way $\gamma$ is defined, we have $c_1 \gamma = c_2$, which is what we claimed.

We can now compute the entire vector $\tilmu_i$ up to scaling, since we know $c_1/c_2$, $c_1/c_3$, and so on. Thus it remains to find the mixture weights $w_i$.  Note that these are all non-negative.  Now from the decomposition, note that for each $i$, we can find the quantity 
\[ C_\ell := w_i \norm{\tilmu_i}^\ell.\]
The trick now is to note that by repeating the entire process above with $\ell$ replaced by $\ell+1$, the conditions of the decomposition theorem still hold, and hence we compute
\[ C_{\ell+1} := w_i \norm{\tilmu_i}^{\ell+1}. \]
Thus taking the ratio $C_{\ell+1}/C_{\ell}$ we obtain $\norm{\tilmu_i}$. This can be done for each $i$, and thus using $C_\ell$, we obtain $w_i$.  This completes the analysis assuming we can obtain $\tilmu_i^{(t)}$ without any error. Please see lemma~\ref{lem:gaussians:weights} for details on how to recover the weights $w_i$ in the presence of errors. This establishes the above claim about recovering the means and weights.

\subsection{Step $3$: Recovering the Variances}
Now that we know the values of $w_i$ and all the means $\tilmu_i$, we show how to recover the variances. This can be done in many ways, and we will outline one which ends up solving a linear system of equations.  Recall that for each Gaussian, the covariance matrix is diagonal (denoted $\Sigma_i$, with $j$th entry equal to $\sigma_{ij}^2$).

Let us show how to recover $\sigma_{i1}^2$ for $1 \le i \le R$.  The same procedure can be applied to the other dimensions to recover $\sigma_{ij}^2$ for all $j$.  Let us divide the set of indices $\{2, 3, \dots, n\}$ into $\ell$ (nearly equal) sets $S_1, S_2, \dots, S_\ell$.  Now consider the expression
\[ \calN_1 = \E [ x(1)^2 (x_{|S_1} \otimes x_{|S_2} \otimes \dots \otimes x_{|S_\ell}) ]. \]
\noindent This can be evaluated as before.  Write $\tilmu_i^{(t)}$ to denote the portion of $\tilmu_i$ restricted to $S_t$, and similarly $\eta_i^{(t)}$ to denote the portion of the noise vector $\eta_i$. This gives
\[ \calN_1 = \sum_i w_i (\tilmu_i(1)^2 + \sigma_{i1}^2) (\tilmu_i^{(1)} \otimes \tilmu_i^{(2)} \otimes \dots \otimes \tilmu_i^{(\ell)}). \]
\noindent Now recall that we {\em know} the vectors $\tilmu_i$ and hence each of the tensors $\tilmu_i^{(1)} \otimes \tilmu_i^{(2)} \otimes \dots \otimes \tilmu_i^{(\ell)}$.  Further, since our $\tilmu_i$ are the perturbed means, our theorem (Theorem~\ref{thm:krl}) about the condition number of Khatri-Rao products implies that the matrix (call it $\calM$) whose columns are the flattened $\prod_t \tilmu_i^{(t)}$ for different $i$, is well conditioned, i.e., has $\sigma_R (\cdot) \ge 1/\poly_\ell (n/\rho)$.  This implies that a system of linear equations $\calM z = z'$ can be solved to recover $z$ up to a $1/\poly_\ell(n/\rho)$ accuracy (assuming we know $z'$ up to a similar accuracy). 

Now using this with $z'$ being the flattened $\calN_1$ allows us to recover the values of $w_i (\tilmu_i(1) + \sigma_{i1}^2)$ for $1 \le i \le R$.  From this, since we know the values of $w_i$ and $\tilmu_i(1)$ for each $i$, we can recover the values $\sigma_{i1}^2$ for all $i$.  As mentioned before, we can repeat this process for other dimensions and recover $\sigma_{ij}^2$ for all $i,j$.

\section{Acknowledgements}
We thank Ryan O'Donnell for suggesting that we extend our techniques for learning mixtures of spherical Gaussians to the more general problem of learning axis-aligned Gaussians. 

%\newpage

\appendix

\section{Stability of the recovery algorithm} \label{asec:stable}

In this section we prove Theorem~\ref{thm:stability}, which shows that the algorithm from section~\ref{sec:review} is actually robust to errors, under Condition~\ref{condition:approx}. This consists of two parts: first proving that the preprocessing step indeed allows us to recover $u_i$ (approximately), and second, that {\sc Decompose} is robust to noise.

\newcommand{\Utilde}{\widetilde{U}}
\subsubsection*{Stability of the preprocessing}
Suppose we are given $T+E$, where $T = \sum_i u_i \ot v_i \ot w_i$, and $E$ is a tensor each of whose entries is $< \epsilon \cdot \poly (1/\kappa, 1/n, 1/\delta)$. Let $\utilde_{j,k}$ be vectors in $\Re^m$ defined as before, and let $\Utilde$ be the $m \times np$ matrix whose columns are $\utilde_{j,k}$ (for different $j,k$). Let $u_i'$ be the projection of $u_i$ onto the span of the top $R$ singular vectors of $\Utilde$. By Claim 4.4 in~\cite{BCV}, we have $\norm{T - \sum_i u_i' \ot v_i \ot w_i}_F < 2\norm{E}_F$, and thus from the robust version of Kruskal's uniqueness theorem~\cite{BCV}, we must have that $u_i'$ and $u_i$ are $\epsilon \cdot \poly(1/\kappa, 1/n, 1/\delta)$ close. Repeating the above along the second mode allows us to move to an $R \times R \times p$ tensor.

%The first observation is that $\Utilde$ must have at least $R$ singular values of magnitude $> \vartheta(1/\kappa, 1/n, 1/\delta)$ (where $\vartheta(\cdot)$ is a polynomial which comes from the Robust Kruskal's theorem~\cite{BCV}). This is because if not, we can project each of the $u_i$ onto this $<R$ dimensional space, thus obtaining a matrix $U'$ which (a) is significantly different from any permutation of $U$ (because the $R$th singular value of any permutation of $U$ is at least $1/Cn\kappa$), but (b) has $\norm{ \sum_i u_i' \ot v_i \ot w_i - \sum_i u_i \ot v_i \ot w_i }_F < (npC^2)\cdot \vartheta(1/\kappa, 1/n, 1/\delta)$. This contradicts the robust uniqueness theorem.

%Now let $\calS$ be the span of the top $R$ singular vectors of $\Utilde$. No unit vector orthogonal to the span of $u_i$ can be in $\calS$, because for such a vector (call it $v$), $\norm{\Utilde v}$ can be bounded by $\norm{E}_F$, and this is $<\vartheta(1/\kappa, 1/n, 1/\delta)$ by assumption. Thus $\calS$ has to be contained in the span of the $u_i$

\subsubsection*{Stability of {\sc Decompose}}
Next, we establish that {\sc Decompose} is stable (in what follows, we have $m=n=R$). Intuitively, {\sc Decompose} is stable provided that the matrices $U$ and $V$ are well-conditioned and the eigenvalues of the matrices that we need to diagonalize are separated. 
%We will assume the following conditions:
%
%\begin{condition}\label{condition:approx}
%
%\begin{enumerate}
%
%\item $\kappa(U), \kappa(V) \geq \kappa$
%
%\item for all $i \neq j$, $\|\frac{w_i}{\|w_i\|} - \frac{w_j}{\|w_j\|} \|_2 \geq \delta$ 
%
%\item for all $i$, $\|u_i\|_2, \|v_i\|_2, \|w_i\|_2 \leq C$
%
%\end{enumerate}
%\end{condition}

The main step in {\sc Decompose} is an eigendecomposition, so first we will establish perturbation bounds. The standard perturbation bounds are known as sin $\theta$ theorems following Davis-Kahan and Wedin. However these bounds hold most generally for the singular value decomposition of an arbitrary (not necessarily symmetric) matrix. We require perturbation bounds for eigen-decompositions of general matrices. There are known bounds due to Eisenstat and Ipsen, however the notion of separation required there is difficult to work with and for our purposes it is easier to prove a direct bound in our setting.

Suppose $M = U D U^{-1}$ and $\widehat{M} = M ( I + E) + F$ and $M$ and $\widehat{M}$ are $n \times n$ matrices. In order to relate the eigendecompositions of $M$ and $\widehat{M}$ respectively, we will first need to establish that the eigenvalues of $M$ are all distinct. We thank Santosh Vempala for pointing out an error in an earlier version. We incorrectly used the Bauer-Fike Theorem to show that $\widehat{M}$ is diagonalizable, but this theorem only shows that each eigenvalue of $\widehat{M}$  is close to some eigenvalue of $M$, but does not show that there is a one-to-one mapping. Fortunately there is a fix for this that works under the same conditions (but again see \cite{GVX} for an earlier, alternative proof that uses a ``homotopy argument"). 

\begin{definition}
Let $\mbox{sep}(D) = \min_{i \neq j} |D_{i,i} - D_{j,j}|$.
\end{definition}

Our first goal is to prove that $\widehat{M}$ is diagonalizable, and we will do this by establishing that its eigenvalues are distinct if the error matrices $E$ and $F$ are not too large. Consider $$U^{-1} ( M ( I + E) + F) U = D + R$$ where $R = U^{-1} (M E + F) U$. We can bound each entry in $R$ by $\kappa(U) (\|M E \|_2 + \|F \|_2)$. Hence if $E$ and $F$ are not too large, the eigenvalues of $D + R$ are close to the eigenvalues of $D$ using Gershgorin's disk theorem, and the eigenvalues of $D + R$ are the same as the eigenvalues of $\widehat{M}$ since these matrices are similar. So we conclude:

%We first apply the Bauer-Fike Theorem to show that the eigenvalues of $\widehat{M}$ are also distinct provided that the perturbation is small enough. Hence $\widehat{M}$ is also diagonalizable. Let us restate the Bauer-Fike Theorem in a more convenient form:

%\begin{theorem}[Bauer-Fike]
%Let $\widehat{\lambda}$ be an eigenvalue of $\widehat{M}$ with eigenvector $\widehat{u}$. Then there is an eigenvalue $\lambda$ of $M$ such that $$|\lambda - \widehat{\lambda}| \leq \kappa (U) \|M \widehat{u} - \widehat{\lambda} \widehat{u}\|_2 \leq \kappa(U) (\|M E u \|_2 + \|F u \|_2).$$
%\end{theorem}

%Hence if the perturbations $E$ and $F$ are small enough, each eigenvalue of $\widehat{M}$ will be closer than $sep(D)/2$ to an eigenvalue of $M$ and consequently $\widehat{M}$ has distinct eigenvalues.

\begin{lemma}
If $\kappa(U) (\|M E \|_2 + \|F \|_2)< \mbox{sep}(D)/(2n)$ then the eigenvalues of $\widehat{M}$ are distinct and it is diagonalizable. 
\end{lemma}

Next we prove that the eigenvectors of $\widehat{M}$ are also close to those of $M$ (this step will rely on $\widehat{M}$ being diagonalizable). This technique is standard in numerical analysis, but it will be more convenient for us to work with relative perturbations (i.e. $\widehat{M} = M ( I + E) + F$) so we include the proof of such a bound for completeness

Consider a right eigenvector $\widehat{u}_i$ of $\widehat{M}$ with eigenvalue $\widehat{\lambda}_i$. We will assume that the conditions of the above corollary are met, so that there is a unique eigenvector $u_i$ of $M$ with eigenvalue $\lambda_i$ which it is paired with. Then since the eigenvectors $\{u_i\}_i$ of $M$ are full rank, we can write $\widehat{u}_i = \sum_j c_j u_j$. Then 
\begin{eqnarray*}
\widehat{M} \widehat{u}_i &=& \widehat{\lambda}_i \widehat{u}_i \\
\sum_j c_j \lambda_j u_j + (ME + F) \widehat{u}_i &=& \widehat{\lambda}_i \widehat{u}_i \\
\sum_j c_j (\lambda_j - \widehat{\lambda}_i) u_j &=& - (ME + F) \widehat{u}_i
\end{eqnarray*}

Now we can left multiply by the $j^{th}$ row of $U^{-1}$; call this vector $w_j^T$. Since $U^{-1} U = I$, we have that $w_j^T u_i = \mathbf{1}_{i = j}$. Hence
$$c_j (\lambda_j - \widehat{\lambda}_i) = - w_j^T (M E + F) \widehat{u}_i$$

\noindent So we conclude:
$$\|\widehat{u}_i - u_i \|_2^2 = 2 dist(\widehat{u}_i, span(u_i))^2 \leq 2 \sum_{j \neq i} \Big ( \frac{(w_j^T (ME + F) \widehat{u}_i)}{|\lambda_j - \widehat{\lambda}_i|} \Big )^2 \leq 8 \sum_{j \neq i} \frac{\|U^{-1} (M E + F ) \widehat{u}_i\|_2^2}{sep(D)^2}$$

\noindent where we have used the condition that $\kappa(U) (\|M E \|_2 + \|F \|_2) < sep(D)/2$ to lower bound the denominator. Furthermore: $\|U^{-1} M E \widehat{u}_i\|_2 = \| D U^{-1} E \widehat{u}_i\|_2 \leq \frac{\sigma_{max}(E) \lambda_{max}(D)}{\sigma_{min}(U)}$ since $\widehat{u}_i$ is a unit vector. 

\begin{theorem}\label{thm:whyisthisnotknown}
If $\kappa(U) (\|M E \|_2 + \|F \|_2) < sep(D)/2$, then $$\|\widehat{u}_i - u_i \|_2 \leq 3 \frac{\sigma_{max}(E) \lambda_{max}(D) + \sigma_{max}(F)}{\sigma_{min}(U)sep(D)}$$
\end{theorem}

Now we are ready to analyze the stability of {\sc Decompose}: Let $T = \sum_{i = 1}^n u_i \otimes v_i \otimes w_i$ be an $n \times n \times p$ tensor that satisfies Condition~\ref{condition:approx}. In our settings of interest we are not given $T$ exactly but rather a good approximation to it, and here let us model this noise as an additive error $E$ that is itself an $n \times n \times p$ tensor.

\begin{claim}\label{claim:separation-evals}
With high probability, $\mbox{sep}(D_a D_b^{-1})  \geq \frac{\delta}{\sqrt{p}}$. 
\end{claim}
\begin{proof}
Fix some $i,j$. The $(i,i)$th entry of $D_a D_b^{-1}$ is precisely $\frac{\iprod{w_i, a}}{\iprod{w_i,b}}$. Note that $\Pr[ |\iprod{w_i, b}| > n \norm{w_i} ]$ is $\exp(-n)$, thus the denominators are all at least $1/(Cn)$ in magnitude with probability $1-\exp(-n)$.

Now given $b$ for which this happens, we have $\frac{\iprod{w_i, a}}{\iprod{w_i,b}} - \frac{\iprod{w_j, a}}{\iprod{w_j,b}} = c_i \iprod{w_i, a} - c_j \iprod{w_j, a}$ where $c_i, c_j$ have magnitude $>1/(Cn)$.  Because $w_i$ has at least a $\delta$ component orthogonal to $w_j$, anti-concentration of Gaussians implies that the difference above is at least $\delta/C^2 n^6$ with probability at least $1-1/n^4$.  Thus we can take a union bound over all pairs.
\end{proof}

We will make crucial use of the following matrix identity:

$$(A + Z)^{-1} = A^{-1} - A^{-1} Z (I + A^{-1} Z)^{-1} A^{-1}$$

\noindent Let $N_a = T_a + E_a$ and $N_b = T_b + E_b$. Then using the above identity we have:
	
$$N_a (N_b)^{-1} = T_a (T_b)^{-1} (I + F) + G$$

\noindent where $F =  -E_b (I + (T_b)^{-1} E_b)^{-1} (T_b)^{-1} $ and $G = E_a (T_b)^{-1}$

\begin{claim}
$\sigma_{max}(F) \leq  \frac{\sigma_{max}(E_b)}{ \sigma_{min}(T_b) - \sigma_{max}(E_b)}$ and $\sigma_{max}(G) \leq \frac{\sigma_{max}(E_a)}{\sigma_{min}(T_b)}$
\end{claim}

\begin{proof}
Using Weyl's Inequality we have $$\sigma_{max}(F)  \leq \frac{\sigma_{max}(E_b)}{1 - \frac{\sigma_{max}(E_b)}{\sigma_{min}(T_b)}} \times \frac{1}{\sigma_{min}(T_b)} = \frac{\sigma_{max}(E_b)}{ \sigma_{min}(T_b) - \sigma_{max}(E_b)}$$ as desired. The second bound is obvious. 
\end{proof}

We can now use Theorem~\ref{thm:whyisthisnotknown} to bound the error in recovering the factors $U$ and $V$ by setting e.g. $M = T_a (T_b)^{-1}$. Additionally, the following claim establishes that the linear system used to solve for $W$ is well-conditioned and hence we can also bound the error in recovering $W$. 

\begin{claim}\label{claim:kappa}
$\kappa(U \odot V) \leq \frac{\min(\sigma_{max}(U), \sigma_{max}(V))}{\max(\sigma_{min}(U), \sigma_{min}(V))} \leq \min(\kappa(U), \kappa(V))$
\end{claim}

\noindent These bounds establish what we qualitatively asserted: {\sc Decompose} is stable provided that the matrices $U$ and $V$ are well-conditioned and the eigenvalues of the matrices that we need to diagonalize are separated. 

\section{K-rank of the Khatri-Rao product}

\subsection{Leave-One-Out Distance}\label{sec:loo}

Recall: we defined the leave-one-out distance in Section~\ref{sec:multiply}. Here we establish that is indeed equivalent to the smallest singular value, up to polynomial factors. In our main proof, this quantity will be much easer to work with since it allows us to translate questions about a set of vectors being well-conditioned to reasoning about projection of each vector onto the orthogonal complement of the others. 

\begin{proof}[Proof of Lemma~\ref{lem:loo}]
Using the variational characterization for singular values: 
$\sigma_{min}(A) = \min_{u, \| u \|_2 = 1} \|A u \|_2$.
Then let $i = \mbox{argmax} |u_i|$. Clearly $|u_i| \geq 1/\sqrt{m}$ since $\|u\|_2 = 1$. Then $
\| A_i + \sum_{j \neq i} A_j \frac{u_j}{u_i}\|_2 =\frac{\sigma_{min}(A)}{u_i} $.
Hence $$ \ell(A) \leq dist(A_i, span\{A_j\}_{j \neq i}) \leq  \frac{\sigma_{min}(A)}{u_i} \leq \sigma_{min}(A) \sqrt{m}$$
Conversely, let $i = \mbox{argmin}_i dist(A_i, span\{A_j\}_{j \neq i})$. Then there are coefficients (with $u_i = 1$) such that $$\| A_i u_i + \sum_{j \neq i} A_j u_j\|_2 = \ell(A).$$ Clearly $\|u\|_2 \geq 1$ since $u_i =1$. And we conclude that $$\ell(A) = \| A_i u_i + \sum_{j \neq i} A_j u_j\|_2 \geq \frac{\| A_i u_i + \sum_{j \neq i} A_j u_j\|_2}{\|u\|_2} \geq \sigma_{min}(A).$$
\end{proof}

\subsection{Proof of Proposition~\ref{prop:matrixcase}} \label{app:prop10}
We now give the complete details of the proof of Proposition~\ref{prop:matrixcase}, that shows how the Kruskal rank multiplies in the smoothed setting for two-wise products. The proof follows by just combining Lemma~\ref{lem:toosystem:create} and Lemma~\ref{lem:toosystem:implies}.

Let $\calU$ be the span of the top $\delta n^2$ singular values of $M$.  Thus $\calU$ is a $\delta n^2$ dimensional subspace of $\R^{n^2}$.  Using Lemma~\ref{lem:toosystem:create} with:
\[ r = \frac{n^{1/2}}{2}, ~~ m = n, ~~ \delta' = \frac{\delta}{n^{1/2}}, \]
we obtain $n \times n$ matrices $M_1, M_2, \dots, M_r$ having the $(\theta, \delta')$-orthogonality property. Note that in this setting, $\delta' m = \frac{n^{1/2}}{2}$.

Thus by applying Lemma~\ref{lem:toosystem:implies}, we have that the matrix $Q(\xtil)$, defined as before, satisfies
\begin{equation}
\label{eq:singvalue-bound}
\Pr_{x} \left[ \sigma_{r/2}\left( Q(\xtil) \right) \ge \frac{\rho \theta}{n ^4} \right] \ge 1- \exp(-r). 
\end{equation}
Now let us consider
\[ \sum_s (\ytil^T M_s \xtil)^2 = \norm{\ytil^T Q(\xtil)}^2. \]
Since $Q(\xtil)$ has many non-negligible singular values (Eq.\eqref{eq:singvalue-bound}), we have (by Fact~\ref{lem:sv:anticonc} for details) that an $\rho$-perturbed vector has a non-negligible norm when multiplied by $Q$. More precisely, $\Pr[ \norm{\ytil^T Q(\xtil)} \ge \rho \theta/n^4] \ge 1-\exp(-r/2)$.  Thus for one of the terms $M_s$, we have $|M_s (\xtil \otimes \ytil)| \ge \rho \theta/n^5$ with probability $\ge 1-\exp(-r/2)$. 

Now this {\em almost} completes the proof, but recall that our aim is to argue about $M(\xtil \otimes \ytil)$, where $M$ is the given matrix.  $\vect(M_s)$ is a vector in the span of the top $\delta n^2$ (right) singular vectors of $M$, and $\sigma_{\delta n^2} \ge \tau$, thus we can write $M_s$ as a combination of the rows of $M$, with each weight in the combination being $\le n/\tau$ (Lemma~\ref{lem:smallcomb}). This implies that for at least one row $M^{(j)}$ of the matrix $M$, we must have \[ \norm{M^{(j)}(\xtil \otimes \ytil} \ge \frac{\theta \rho \tau}{n^6} = \frac{\rho \tau}{n^{O(1)}}. \]
(Otherwise we have a contradiction). This completes the proof.

\qed

Before we give the complete proofs of the two main lemmas regarding ordered $(\theta,\delta)$ orthogonal systems (Lemma~\ref{lem:toosystem:create} and Lemma~\ref{lem:toosystem:implies}), we start with a simple lemma about top singular vectors of matrices, which is very useful to obtain linear combinations of small length. 
\begin{lemma}[Expressing top singular vectors as small combinations of columns] \label{lem:smallcomb} 
Suppose we have a $m \times n$ matrix $M$ with $\sigma_t(M) \ge \eta$, and let $v_1, v_2, \dots v_t \in \R^m$ be the top $t$ left-singular vectors of $M$. Then these top $t$ singular vector can be expressed using \emph{small} linear combinations of the columns $\set{M(i)}_{i \in [n]}$ i.e.
\begin{align*}
\forall k \in [t], ~\exists \set{\alpha_{k,i}}_{i \in [n]} ~\text{such that }&~ v_k = \sum_{i \in [n]} \alpha_{k,i} M(i) \\
&\text{and} \sum_i \alpha_{k,i}^2 \le 1/\eta^2  
\end{align*}    
\end{lemma}
\begin{proof}
Let $\ell$ correspond to the number of non-zero singular values of $M$. 
Using the SVD, there exists matrices $V \in \R^{m \times \ell}, U \in \R^{n \times \ell}$ with orthonormal columns (both unitary matrices), and a diagonal matrix $\Sigma \in \R^{\ell \times \ell}$ such that $M=V \Sigma \transpose{U}$. Since the $n \times \ell$ matrix $V=M (U \Sigma^{-1})$, the $t$ columns of $V$ corresponding to the top $t$ singular values ($\sigma_{t}(M) \ge \eta$) correspond to linear combinations which are small i.e. $\forall k \in [t], ~ \norm{\alpha_k} \le 1/\eta$.
\end{proof}

%%%%%%%%%%%%%%%%%%%%%%%%%%%%%%%%%%%%%%%%%%%%%%%%%%%%%%%%%%%%%%%%%%%%%%%%%%%%%%%%%%%%%%%%%%%%%%%%%%%%%%%%%%%
\subsection{Constructing the $(\theta,\delta)$-Orthogonal System (Proof of Lemma~\ref{lem:toosystem:create})} \label{app:construction}
Let $\calV$ be a subspace of $\R^{n \cdot m}$, with its co-ordinates indexed by $[n]\times [m]$.  
Further,remember that the vectors in $\R^{n \cdot m}$ are also treated as matrices of size $n \times m$. 
%To construct an ordered $(\theta,\delta)$-orthogonal system of $r'$ matrices $\set{Q_j}$, we need to pick columns (indexed by $[m]$) which can have many linearly-independent choices for the $r'$ different matrices. Hence, we want to choose columns such that $\calV$ projected onto these columns span large dimensional subspaces of $\R^n$ (in a robust sense). This is captured by the robust dimension of column projections.    
%
%\begin{definition}[Dimension of projections]\label{def:dimproj}
%For a subspace $\calV$ of $\R^{n\cdot m}$, we define its robust dimension $\dimt{i}{\calV}$ to be
%\begin{align*}
%\dimt{i}{\calV}=\max_d ~\text{s.t.} ~& \exists ~\text{orthonormal } u_1,u_2,\dots, u_d \in \R^n \text{ and } M_1,M_2, \dots, M_d \in \calV \\ 
%&~\text{with } \forall t \in [d], \norm{M_t}\le \tau ~\text{and } u_t = M_t(i)  .
%\end{align*}
%\end{definition}

%%\newcommand{\rhotau}{\frac{1}{\sqrt{p_2}}}
%\begin{lemma}\label{lem:blocks}
%In any subspace $\calV$ in $\R^{p_1 \cdot p_2}$ of dimension $\dim(\calV)$ for any $\tau \ge \sqrt{p_2}$, we have 
%\begin{equation}
%\sum_{i \in [p_2]} \dimt{i}{\calV} \ge \dim(\calV)
%\end{equation} 
%\end{lemma}
We now give the complete proof of lemma~\ref{lem:blocks} that shows that the average robust dimension of column projections is large if the dimension of $\calV$ is large . 

\begin{proof}[Proof of Lemma~\ref{lem:blocks}]
Let $d=\dim(\calV)$. Let $B$ be a $p_1 p_2 \times d$ matrix composed of a orthonormal basis (of $d$ vectors) for $\calV$ i.e. the $j^{th}$ column of $B$ is the $j^{th}$ basis vector ($j \in [d]$) of $\calV$. Clearly $\sigma_d (B)=1$. \\   
For $i \in [p_2]$,  let $B_i$ be the $p_1 \times d$ matrix obtained by projecting the columns of $B$ on just the rows given by $[p_1] \times i$. Hence, $B$ is obtained by just concatenating the columns as $\transpose{B}= \left[ \transpose{B_1} \| \transpose{B_2} \| \dots \| \transpose{B_p} \right]$. 
Finally, let $d_i=\max t$ such that $\sigma_t(B_i) \ge \rhotau$. 

We will first show that $\sum_i d_i \ge d$. Then we will show that $\dimt{i}{\calV} \ge d_i$ to complete our proof.\\
Suppose for contradiction that $\sum_{i \in [p_2]} d_i < d$. Let $\calS_i$ be the $(d-d_1)$-dimensional subspace of $\R^d$ spanned by the last $(d-d_1)$ right singular vectors of $B_i$. Hence, 
$$\text{ for unit vectors} ~\alpha \in \calS_i \subseteq \R^d,~ \norm{B_i \alpha} < \rhotau .$$
Since, $d-\sum_{i \in [p_2]} d_i >0$, there exists at least one unit vector $\alpha \in \bigcap_i \calS_i^{\perp}$. Picking this unit vector $\alpha \in \R^d$, we have $\norm{B \alpha}_2^2 = \sum_{i \in [p_2]} \norm{B_i \alpha}_2^2 
< p_2\cdot (\rhotau)^2 < 1$. This contradicts $\sigma_d(B) \ge 1$  

%Hence $\sum_i d_i \ge d$. 
To establish the second part, consider some $B_i$ ($i \in [p_2]$). We pick $d_i$ orthonormal vectors $\in \R^{p_1}$ corresponding to the top $d_i$ left-singular vectors of $B_i$. By using Lemma~\ref{lem:smallcomb}, we know that each of these $j \in [d_i]$ vectors can be expressed as a small combination $\vec{\alpha_j}$ of the columns of $B_i$ s.t. $\norm{\vec{\alpha_j}} \le \sqrt{p_2}$.
Further, if we associate with each of these $j \in [d_i]$ vectors, the vector $w_j \in \R^{(p_1 p_2)}$ given by the same combination $\vec{\alpha_j}$ of the columns of $B$, we see that $\norm{w_j} \le \sqrt{p_2}$ since the columns of the matrix $B$ are orthonormal.  
\end{proof}

%\end{proof}
%%%%%%%%%%%%%%%%%%%%%%%%%%%%%%%%%%%%%%%%%%%%%%%%%%%%%%%%%%%%%%%%%%%%%%%%%%%%%%%%%%%%%%%%%%%%%%

%\TODO{Maybe combine all these useful linear algebraic lemmas?}
%\subsubsection{Some Useful Lemmas}
%
%\begin{lemma}[Characterizing singular values via concentration for random vectors]\label{lem:random:sv}
%For an $\ell \times n$ matrix $M$,
%\begin{equation}
%\Pr_{z \in \calN(0,1)^\ell}\left[ \norm{z^{t} M} \ge \frac{1}{n^4} \right] \ge 1-n^{-100\ell} \implies \sigma_{\Omega(\ell)} \ge 1/n^{6} ~\text{ w.p. } 1 - \exp(-\tilde{\Omega}(\ell)) 
%\end{equation} 
%\end{lemma}

\subsection{Implications of Ordered $(\theta,\delta)$-Orthogonality: Details of Proof of Lemma~\ref{lem:toosystem:implies}} \label{app:too-system-implies}

Here we show some auxiliary lemmas that are used in the Proof of Lemma~\ref{app:too-system-implies}.

\begin{claim}\label{claim:theta-orthog}
Suppose $v_1, v_2, \dots, v_m$ are a set of vectors in $\Re^n$ of length $\le 1$, having the $\theta$-orthogonal property.  Then we have
\begin{enumerate}
\item[(a)] For $g\sim \calN(0,1)^n$, we have $\sum_i \iprod{v_i, g}^2 \ge \theta^2/2$ with probability $\ge 1-\exp(-\Omega(m))$,
\item[(b)] For $g\sim \calN(0,1)^m$, we have $\norm {\sum_i g_i v_i}^2 \ge \theta^2/2$ with probability $\ge 1-\exp(-\Omega(m))$.
\end{enumerate}
Furthermore, part (a) holds even if $g$ is drawn from $u+g'$, for any fixed vector $u$ and $g' \sim \calN(0,1)^n$.
\end{claim}
\begin{proof}
First note that we must have $m \le n$, because otherwise $\{v_1, v_2, \dots, v_m\}$ cannot have the $\theta$-orthogonal property for $\theta>0$.  For any $j \in [m]$, we claim that
\begin{equation}\label{eq:subclaim}
\Pr[ ( \iprod{v_j, g}^2 < \theta^2/2 ) ~|~ v_1, v_2, \dots, v_{j-1} ] < 1/2.
\end{equation}
To see this, write $v_j = v_j' + v_j^\perp$, where $v_j^\perp$ is orthogonal to the span of $\{v_1, v_2, \dots, v_{j-1}\}$.  Since $j \in I$, we have $\norm{v_j^\perp} \ge \theta$.  Now given the vectors $v_1, v_2, \dots, v_{j-1}$, the value $\iprod{v_j', g}$ is fixed, but $\iprod{ v_j^\perp, g}$ is distributed as a Gaussian with variance $\theta^2$ (since $g$ is a Gaussian of unit variance in each direction).

Thus from a standard anti-concentration property for the one-dimensional Gaussian, $\iprod{v_j, g}$ cannot have a mass $>1/2$ in any $\theta^2$ length interval, in particular, it cannot lie in $[-\theta^2/2, \theta^2/2]$ with probability $>1/2$.  This proves Eq.~\eqref{eq:subclaim}.  Now since this is true for any conditioning $v_1, v_2, \dots, v_{j-1}$ and for all $j$, it follows (see Lemma~\ref{lem:basic-prob1} for a formal justification) that
\[ \Pr[ \iprod{v_j, g}^2 < \theta^2/2 \text{ for all } j] < \frac{1}{2^{m}} < \exp(-m/2). \]
This completes the proof of the claim, part (a).  Note that even if we had $g$ replaced by $u+g$ throughout, the anti-concentration property still holds (we have a shifted one-dimensional Gaussian), thus the proof goes through verbatim.

Let us now prove part (b).  First note that if we denote by $M$ the $n \times m$ matrix whose columns are the $v_i$, then part (a) deals with the distribution of $g^T MM^T g$, where $g \sim \calN(0,1)^n$.  Part (b) deals with the distribution of $g^T M^T M g$, where $g \sim \calN(0,1)^m$.  But since the eigenvalues of $MM^T$ and $M^TM$ are precisely the same, due to the rotational invariance of Gaussians, these two quantities are distributed exactly the same way.  This completes the proof.
\end{proof}

\begin{lemma}\label{lem:basic-prob1}
Suppose we have random variables $X_1, X_2, \dots, X_r$ and an event $f(\cdot)$ which is defined to occur if its argument lies in a certain interval (e.g. $f(X)$ occurs iff $0<X<1$). Further, suppose we have $\Pr[f(X_1)] \le p$, and $\Pr[f(X_i)| X_1, X_2, \dots, X_{i-1}] \le p$ for all $X_1, X_2, \dots, X_{i-1}$. Then
\[ Pr[f(X_1) \wedge f(X_2) \wedge \dots \wedge f(X_r)] \le p^r. \]
\end{lemma}
%
%\begin{lemma}\label{lem:sing-value-prob}
%Let $M$ be a $t \times t$ matrix with spectral norm $\le 1$.  Suppose $M$ has at most $r$ singular values of magnitude $> \tau$.  Then for $g \sim \calN(0,1)^t$, we have
%\[ \Pr[ \norm{Mg}_2^2 < 4t\tau^2 + \frac{t}{n^{2c}} ] \ge \frac{1}{n^{cr}} - \frac{1}{2^t}. \]
%\end{lemma}
%\begin{proof}
%Let $u_1, u_2, \dots, u_r$ be the singular vectors corresponding to value $> \tau$.  Consider the event that $g$ has a projection of length $< 1/n^c$ onto $u_1, u_2, \dots, u_r$.  This has probability $\ge \frac{1}{n^{cr}}$, by anti-concentration properties of the Gaussian (and because $\calN(0,1)^t$ is rotationally invariant).  For any such $g$, we have
%\begin{align*}
%\norm{ Mg}_2^2 &= \sum_{i=1}^r \iprod{g, u_i}^2 + \tau^2 \norm{g}^2 \\
%&\le \frac{r}{n^{2c}} + \tau^2 \norm{g}_2^2.
%\end{align*} 
%\end{proof}

%%%%%%%%%%%%%%%%%%%%%%%%%%%%%%%%%%%%%%%%%%%%%%%%%%%%%%%%%%%%%%%%%%%%%%%%%%%%%%%%%%%%%%%%%%%%%%%%%%%
%\input{higher_order}

\newcommand{\nrm}[2]{\norm{#1}_{#2}}
\newcommand{\minl}{L_{\min}}
\newcommand{\maxl}{L_{\max}}
\newcommand{\up}{\tilde{u}_{\perp}}
\newcommand{\vp}{\tilde{v}_{\perp}}

\newcommand{\mucap}{\widehat{\mu}}
%%%Appendix
\section{Applications to Mixture Models}

\subsection{Sampling Error Estimates for Multi-view Models} \label{app:sampling:mv}

In this section, we show error estimates for $\ell$-order tensors obtained by looking at the $\ell^{th}$ moment of the multi-view model.

\begin{lemma}[Error estimates for Multiview mixture model]\label{lem:sampling:mv}
For every $\ell \in \N$, suppose we have a multi-view model, with parameters $\{w_r\}_{r \in [R]}$ and $\{\spc{M}{j}\}_{j \in [\ell]}$,  the $n$ dimensional sample vectors $\spc{x}{j}$ have $\norm{\spc{x}{j} }_\infty \le 1$. Then,
for every $\eps>0$, there exists $N =O(\eps^{-2}\sqrt{\ell \log n})$ such that \\
if $N$ samples $\{\spc{x(1)}{j}\}_{j \in [\ell]},\{\spc{x(2)}{j}\}_{j \in [\ell]},\dots,\{\spc{x(N)}{j}\}_{j \in [\ell]}$ are generated, then with high probability
\begin{equation}\label{eq:sampleerror:topic}
\nrm{\E{\spc{x}{1} \otimes \spc{x}{2} \otimes \dots \spc{x}{\ell} }- \frac{1}{N} \left(\sum_{t \in [N]} \spc{x(t)}{1} \otimes \spc{x(t)}{2} \otimes \spc{x(t)}{\ell} \right)}{\infty} < \eps 
\end{equation}
\end{lemma}
\begin{proof}
We first bound the $\| \cdot\|_{\infty}$ norm of the difference of tensors i.e. we show that
$$\forall \{i_1,i_2,\dots,i_\ell\}\in [n]^\ell, \abs{\E{\prod_{j \in [\ell]} \spc{x}{j}_{i_j} }-\frac{1}{N}\left(\sum_{t \in [N]} \prod_{j \in [\ell]} \spc{x(t)}{j}_{i_j}\right)} < \eps/n^{\ell/2}.$$
Consider a fixed entry $(i_1,i_2, \dots, i_\ell)$ of the tensor. 

Each sample $t \in [N]$ corresponds to an independent random variable with a bound of $1$. Hence, we have a sum of $N$ bounded random variables. By Bernstein bounds, probability for \eqref{eq:sampleerror:topic} to not occur $\exp\left(-\frac{\left(\eps n^{-\ell/2}\right)^2 N^2}{2N}\right)=\exp\left(-\eps^2 N/\left(2n^\ell\right)\right)$. We have $n^\ell$ events to union bound over. Hence $N=O(\eps^{-2} n^\ell \sqrt{\ell \log n})$ suffices. Note that similar bounds hold when the $\spc{x}{j} \in \R^n$ are generated from a multivariate gaussian.
\end{proof}

\subsection{Error Analysis for Multi-view Models}
\begin{lemma}\label{lem:tensoring}
Suppose $\nrm{u \otimes v - u' \otimes v'}{F} < \delta$, and $\minl \le \norm{u}, \norm{v}, \norm{u'}, \norm{v'} \le \maxl$,\\ with $\delta <\frac{\min\{ \minl^2 ,1 \}}{(2\max\{\maxl,1\})}$. If $u=\alpha_1 u'+\beta_1 \up$ and $v=\alpha_2 v'+\beta_2 \vp$, where $\up$ and $\vp$ are unit vectors orthogonal to $u', v'$ respectively, then we have
$$ |1-\alpha_1 \alpha_2|< \delta/ \minl^2 ~\text{ and } \quad \beta_1 < \sqrt{\delta},~\beta_2 < \sqrt{\delta}.$$
\end{lemma}
\begin{proof}
We are given that $u=\alpha_1 u' +\beta_1 \up$ and $v=\alpha_2 v' + \beta_2 \vp$. Now, since the tensored vectors are close
\begin{align}
\nrm{u \otimes v - u' \otimes v'}{F}^2 &< \delta^2 \notag\\  
\nrm{(1-\alpha_1 \alpha_2) u'\otimes v' + \beta_1 \alpha_2 \up \otimes v'+ \beta_2 \alpha_1 u' \otimes \vp + \beta_1\beta_2 \up \otimes \vp}{F}^2 &< \delta^2 \notag\\
\minl^4 (1-\alpha_1 \alpha_2)^2 + \beta_1^2 \alpha_2^2 \minl^2 +\beta_2^2 \alpha_1 ^2 \minl^2 + \beta_1^2 \beta_2^2 &< \delta^2 \label{eq:tensoring:lb}
\end{align}
This implies that $|1-\alpha_1 \alpha_2| < \delta/ \minl^2$ as required.

Now, let us assume $\beta_1> \sqrt{\delta}$. This at once implies that $\beta_2 < \sqrt{\delta}$. 
Also 
\begin{align*}
\minl^2 \le \norm{v}^2 &= \alpha_2^2 \norm{v'}^2 + \beta_2^2  \\
\minl^2 -\delta &\le \alpha_2^2 \maxl^2 \\
\text{ Hence, }\quad \alpha_2 &\ge \frac{\minl}{2\maxl} 
\end{align*}
Now, using \eqref{eq:tensoring:lb}, we see that $\beta_1 < \sqrt{\delta}$.
\end{proof}

\subsection{Sampling Error Estimates for Gaussians} \label{app:sampling:mv}

\newcommand{\mmax}{B}
\begin{lemma}[Error estimates for Gaussians]\label{lem:sampling:gaussians}
Suppose $x$ is generated from a mixture of $R$-gaussians with means $\{\mu_r\}_{r \in [R]}$ and covariance $\Sigma_i$ that is diagonal , with the means satisfying $\norm{\mu_r} \le \mmax$. Let $\sigma= \max_i \sigma_{\max}(\Sigma_i)$\\
For every $\eps>0, \ell \in \N$, there exists $N =\Omega( \poly(\frac{1}{\eps})), \sigma^2, n,R)$ such that 
if $x^{(1)},x^{(2)},\dots,x^{(N)} \in R^n$ were the $N$ samples, then
\begin{equation}\label{eq:sample:errorbound}
\forall \{i_1,i_2,\dots,i_\ell\}\in [n]^\ell, \abs{\E{\prod_{j \in [\ell]} x_{i_j} }-\frac{1}{N}\left(\sum_{t \in [N]} \prod_{j \in [\ell]} x^{(t)}_{i_j}\right)} < \eps.
\end{equation}
In other words,
\[ \nrm{\E{x^{\otimes \ell} }- \frac{1}{N} \big(\sum_{t \in [N]} (x^{(t)})^{\otimes \ell}\big)}{\infty} < \eps  \]
\end{lemma}
\begin{proof}
Fix an element $(i_1,i_2,\dots,i_\ell)$ of the $\ell$-order tensor. Each point $t \in [N]$ corresponds to an i.i.d random variable $Z^t=x^{(t)}_{i_1} x^{(t)}_{i_2} \dots x^{(t)}_{\ell}$. We are interested in the deviation of the sum $S=\frac{1}{N}\sum_{t \in [N]} Z^t$. Each of the i.i.d rvs has value $Z=x_{i_1} x_{i_2} \dots x_{\ell}$. Since the gaussians are axis-aligned and each mean is bounded by $\mmax$, $|Z|< (\mmax+t \sigma)^\ell$ with probability $O\left(\exp(-t^2/2)\right)$. Hence, by using standard sub-gaussian tail inequalities, we get
$$\Pr{\abs{S-\E{z}} > \eps} < \exp\left(-\frac{\eps^2 N}{(M+\sigma \ell \log n)^\ell}\right)$$
Hence, to union bound over all $n^{\ell}$ events $N=O\left(\eps^{-2} (\ell \log n M)^\ell\right)$ suffices.
\end{proof}

\subsection{Recovering Weights in Gaussian Mixtures}
We now show how we can approximate upto a small error the weight $w_i$ of a gaussian components in a mixture of gaussians, when we have good approximations to $w_i \mu_i^{\otimes \ell}$ and $w_i \mu_i^{\otimes (\ell-1)}$.

\begin{lemma}[Recovering Weights] \label{lem:gaussians:weights}
For every $\delta'>0,w>0,\minl>0,\ell \in \N$, $\exists \delta =\Omega\big(\frac{\delta_1 w^{1/(\ell-1)}}{\ell^2 \minl}\big)$ such that, if $\mu \in \R^n$ be a vector with length $\norm{\mu} \ge \minl$, and suppose
\[ \norm{v-w^{1/\ell} \mu} < \delta \quad \text{ and } \norm{u-w^{1/(\ell-1)} \mu} < \delta. \label{eq:weights:1}\]
Then, 
\begin{equation}
\qquad \abs{\left(\frac{\abs{\iprod{u,v}}}{\norm{u}}\right)^{\ell(\ell-1)} - w } < \delta'
\end{equation}
\end{lemma}
\begin{proof}
From \eqref{eq:weights:1} and triangle inequality, we see that 
$$\norm{w^{-1/\ell} v-w^{-1/(\ell-1)}u} \le \delta(w^{-1/(\ell)}+w^{-1/(\ell-1)})= \delta_1.$$
Let $\alpha_1=w^{-1/(\ell-1)}$ and $\alpha_2=w^{-1/\ell}$.
Suppose $v=\beta u + \eps \up$ where $\up$ is a unit vector perpendicular to $u$. Hence $\beta=\iprod{v,u}/\norm{u}$.
\begin{align*}
\norm{\alpha_1 v- \alpha_2 u}^2 &= \norm{(\beta \alpha_1-\alpha_2 )u +\alpha_1\eps \up}< \delta_1^2\\
(\beta \alpha_1 - \alpha_2)^2 \norm{u}^2 + \alpha_1^2 \eps^2 &\le \delta_1^2\\
\abs{\beta - \frac{\alpha_2}{\alpha_1}} &< \frac{\delta_1}{ \minl}
\end{align*}
Now, substituting the values for $\alpha_1,\alpha_2$, we see that
$$\abs{\beta - w^{\frac{1}{(\ell-1)} - \frac{1}{\ell}} } < \frac{\delta_1}{\minl}.$$
\begin{align*}
\abs{\beta - w^{1/(\ell(\ell-1))}} &< \frac{\delta}{w^{1/(\ell-1)} \minl}\\ 
\abs{\beta^{\ell (\ell-1)} - w } &\le \delta' \quad \text{when } \delta \ll \frac{\delta' w^{1/(\ell-1)}}{\ell^2 \minl}
\end{align*}
\end{proof}

\end{document}